\documentclass[11pt]{article}
 \usepackage[round]{natbib}
 \usepackage{amssymb}
\usepackage{amsmath}
\usepackage{subfigure}
\usepackage{fullpage}
\usepackage{amsthm}
\usepackage{graphicx}
\usepackage{xcolor}
\usepackage{algorithm} 
\usepackage{algpseudocode}
\algnewcommand\algorithmicinput{\textbf{INPUT:}}
\algnewcommand\INPUT{\item[\algorithmicinput]}
\algnewcommand\algorithmicoutput{\textbf{OUTPUT:}}
\algnewcommand\OUTPUT{\item[\algorithmicoutput]}

\usepackage{hyperref}[]
\hypersetup{
    colorlinks=true,
    linkcolor=blue,
    filecolor=magenta,      
    urlcolor=cyan,
      citecolor=blue,
    }

\usepackage{cleveref}

\newcommand{\lc}{ \lceil} 
\newcommand{\rc}{ \rceil } 
\newcommand{\kse}{\kappa_{\max}^{s,e}}
\newcommand{\dint}{\,\mathrm{d}}

\DeclareMathOperator*{\argmax}{arg\,max}

\DeclareMathOperator*{\loglog}{log\, log}

\newcommand{\lt}{ {\mathcal L^2}}

\newcommand{\hr}{{ \mathcal H^{r} } }
\newtheorem{assumption}{Assumption}
\newtheorem{lemma}{Lemma}
\newtheorem{proposition}{Proposition}
\newtheorem{corollary}{Corollary}
\newtheorem{definition}{Definition}
\newtheorem{remark}{Remark}
\newtheorem{theorem}{Theorem}
\newtheorem{example}{Example}

\title{Change-point Detection for  Sparse and Dense Functional Data in General Dimensions}

\author{Carlos Misael Madrid Padilla$^1$ \and Daren Wang$^2$ \and  Zifeng Zhao$^3$ \and  Yi Yu$^4$}
\date{%
    $^1$Department of Mathematics, University of Notre Dame\\%
    $^2$Department of ACMS, University of Notre Dame\\%
    $^3$Mendoza College of Business, University of Notre Dame\\%
    $^4$Department of Statistics, University of Warwick\\[2ex]%
    \today
}

\begin{document}

\maketitle

\begin{abstract}
We study the problem of change-point detection and localisation for functional data sequentially observed on a general $d$-dimensional space, where we allow the functional curves to be either sparsely or densely sampled. Data of this form naturally arise in a wide range of applications such as biology, neuroscience, climatology and finance. To achieve such a task, we propose a kernel-based algorithm named functional seeded binary segmentation (FSBS). FSBS is computationally efficient, can handle discretely observed functional data, and is theoretically sound for heavy-tailed and temporally-dependent observations. Moreover, FSBS works for a general $d$-dimensional domain, which is the first in the literature of change-point estimation for functional data.  We show the consistency of FSBS for multiple change-point estimation and further provide a sharp localisation error rate, which reveals an interesting phase transition phenomenon depending on the number of functional curves observed and the sampling frequency for each curve. Extensive numerical experiments illustrate the effectiveness of FSBS and its advantage over existing methods in the literature under various settings. A real data application is further conducted, where FSBS localises change-points of sea surface temperature patterns in the south Pacific attributed to El Ni\~{n}o.
\end{abstract}

\section{Introduction}
Recent technological advancement has boosted the emergence of functional data in various application areas, including neuroscience \citep[e.g.][]{dai2019age,petersen2019frechet}, finance \citep[e.g.][]{fan2014functional}, transportation \citep[e.g.][]{chiou2014functional}, climatology \citep[e.g.][]{bonner2014modeling, fraiman2014detecting} and others. We refer the readers to \cite{wang2016functional} - a comprehensive review, for recent development of statistical research in functional data analysis.

In this paper, we study the problem of change-point detection and localisation for functional data, where the data are observed sequentially as a time series and the mean functions are piecewise stationary, with abrupt changes occurring at unknown time points.  To be specific, denote $\mathcal{D}$ as a general $d$-dimensional space that is homeomorphic to $[0,1]^d$, where $d \in \mathbb{N}^+$ is considered as arbitrary but fixed.  We assume that the observations $\{(x_{t, i}, y_{t, i})\}_{t = 1,i = 1}^{T,n} \subseteq \mathcal{D} \times \mathbb{R}$ are generated based on 
\begin{align}\label{eq_model}
    y_{t,i}  = f^*  _t (x_{t,i} ) + \xi_t (x_{t,i } ) + \delta_{t, i }, \text{ for } t=1,\ldots, T  \text{ and } i = 1, \ldots, n .
\end{align}
In this model, $\{x_{t,i}\}_{t=1,i=1}^{T,n}\subseteq \mathcal{D}$ denotes the discrete grids where the (noisy) functional data $\{y_{t,i}\}_{t=1,i=1}^{T,n}\subseteq \mathbb{R}$ are observed, $\{f_t^*: \mathcal{D}\to \mathbb R\}_{t=1}^T$ denotes the deterministic mean functions, $\{\xi_t: \mathcal{D}\to \mathbb R\}_{t=1}^T$ denotes the functional noise and $\{\delta_{t,i}\}_{t=1,i=1}^{T,n}\subseteq \mathbb{R}$ denotes the measurement error. We refer to \Cref{assume: model assumption} below for detailed technical conditions on the model. 

To model the unstationarity of sequentially observed functional data which commonly exists in real world applications, we assume that there exist $K \in \mathbb{N}$ change-points, namely $0 =\eta_0 < \eta_1 < \cdots < \eta_K  <  \eta_{K+1} = T$, satisfying that $f^*_t \neq f^*_{t+1}$, if and only if $t \in \{\eta_k\}_{k = 1}^K$.  Our primary interest is to accurately estimate $\{\eta_k\}_{k=1}^K$.

Due to the importance of modelling unstationary functional data in various scientific fields, this problem has received extensive attention in the statistical change-point literature, see e.g.~\cite{aue2009estimation}, \cite{berkes2009detecting}, \cite{hormann2010weakly}, \cite{zhang2011testing}, \cite{aue2018detecting} and \cite{dette2020testing}.  Despite the popularity, we identify a few limitations in the existing works. Firstly, both the methodological validity and theoretical guarantees of all these papers require fully observed functional data without measurement error, which may not be realistic in practice.  Secondly, most existing works focus on the single change-point setting and to our best knowledge, there is no consistency result of multiple change-point estimation for functional data.  Lastly but most importantly, existing algorithms only consider functional data with support on $[0,1]$ and thus are not applicable to functional data with multi-dimensional domain, a type of data frequently encountered in neuroscience and climatology.

In view of the aforementioned three limitations, in this paper, we make several theoretical and methodological contributions, summarized below.

$\bullet$ In terms of methodology, our proposed kernel-based change-point detection algorithm, functional seeded binary segmentation (FSBS), is computationally efficient, can handle discretely observed functional data contaminated with measurement error, and allows for temporally-dependent and heavy-tailed data.  FSBS, in particular, works for a general $d$-dimensional domain with arbitrary but fixed $d \in \mathbb{N}^+$.  This level of generality is the first time seen in the literature.
    
$\bullet$ In terms of theory, we show that under standard regularity conditions, FSBS is consistent in detecting and localising multiple change-points.  We also provide a sharp localisation error rate, which reveals an interesting phase transition phenomenon depending on the number of functional curves observed $T$ and the sampling frequency for each curve $n$.  To the best of our knowledge, the theoretical results we provide in this paper are the sharpest in the existing literature.
    
$\bullet$ A striking case we handle in this paper is that each curve is only sampled at one point, i.e.~$n = 1$.  To the best of our knowledge, all the existing functional data change-point analysis papers assume full curves are observed.  We not only allow for discrete observation, but carefully study  this most extreme sparse case $n = 1$ and provide consistent localisation of the change-points.

$\bullet$ We conduct extensive numerical experiments on simulated and real data. The result further supports our theoretical findings, showcases the advantages of FSBS over existing methods and illustrates the practicality of FSBS.
    
$\bullet$ A byproduct of our theoretical analysis is new theoretical results on kernel estimation for functional data under temporal dependence and heavy-tailedness.  This set of new results \emph{per se} are novel, enlarging the toolboxes of functional data analysis.

\textbf{Notation and definition.}  For any function $f:\, [0,1]^d \to \mathbb R$ and for $1 \leq p < \infty$, define   $\|f\|_p = (\int_{[0,1]^d} |f(x)|^p \dint x)^{1/p}$ and for $p = \infty$, define $\|f\|_\infty = \sup_{x\in{[0,1]^d}}\vert f(x)\vert$.  Define     $\mathcal{L}_p = \{f: \, [0, 1]^d \to \mathbb{R}, \, \|f\|_p<\infty \}$.  For any vector $s = (s_1, \ldots, s_d)^{\top} \in{\mathbb{N}^d}$, define $\vert s\vert = \sum_{i = 1}^d s_i$, $s!=s_1!\cdots s_d!$ and the associated partial differential operator $D^s= \frac{\partial^{\vert s\vert}}{\partial x_1^{s_1}\cdots  \partial x_{d}^{s_d}}$.
For $\alpha>0$, denote $\lfloor\alpha\rfloor $ to be the largest integer   smaller than $\alpha$. For any function $f:\, [0,1]^d \to \mathbb R $ that is 
$ \lfloor\alpha\rfloor $-times continuously differentiable at point $x_0$, denote by $f_{x_0}^\alpha $ its Taylor polynomial of degree $ \lfloor\alpha\rfloor $ at $x_0$, which is defined as $f_{x_0}^\alpha(x) = \sum_{|s| \le \lfloor\alpha\rfloor  } \frac{(x-x_0)^s  }{s!} D^s f(x_0).$ For a constant $L>0$, let $\mathcal{H}^{\alpha}(L)$ be the set  of functions $f:\, [0,1]^d \to \mathbb R $ such that $f$ is $\lfloor\alpha\rfloor$-times differentiable for all  $  x  \in [0,1]^d$  and satisfy $ |f(x) - f_{x_0}^\alpha(x)  | \le L| x-x_0|^ \alpha$, for all $x, x_0\in [0,1]^d$. Here $|x-x_0| $ is  the Euclidean distance between $x, x_0 \in \mathbb R^d$. In non-parametric statistical literature, $\mathcal{H}^{\alpha}(L)$ are often referred to as the  class of H\"{o}lder  smooth functions. We refer the interested readers to \cite{rigollet2009optimal} for more detailed discussion on H\"{o}lder smooth functions.

For two positive sequences $\{a_n\}_{n\in \mathbb N^+ }$ and $\{b_n\}_{n\in \mathbb N ^+ }$, we write $a_n = O(b_n)$ or $a_n\lesssim b_n$ if $a_n\le Cb_n$ with some constant $C > 0$ that does not depend on $n$, and $a_n = \Theta(b_n)$ or $a_n\asymp b_n$ if $a_n = O(b_n)$ and $b_n = O(a_n)$. 
 

\section{Functional seeded binary segmentation}

\subsection{Problem formulation}

Detailed model assumptions imposed on model \eqref{eq_model} are collected in \Cref{assume: model assumption}. For notational simplicity, without loss of generality, we set the general $d$-dimensional domain $\mathcal{D}$ to be $[0,1]^d$, as the results apply to any $\mathcal{D}$ that is homeomorphic to $[0,1]^d$.

\begin{assumption}\label{assume: model assumption}
The data $\{(x_{t, i}, y_{t, i})\}_{t = 1,i = 1}^{T,n} \subseteq [0,1]^d \times \mathbb{R}$ are generated based on model \eqref{eq_model}.

\noindent {\bf a.} (Discrete grids) The grids $\{ x_{t,i}  \}_{t = 1,i = 1}^{T,n} \subseteq [0, 1]^d$ are independently sampled from a common   density function $u:\, [0,1]^d \to \mathbb R$. In addition, there exist constants  $r>0$ and $L>0$ such that   $ u \in  \mathcal{H}^r (L)  $ and that  $\inf_{x\in [0,1]^d } u(x) \ge \tilde{c}$ with an absolute constant $\tilde{c}>0$.

\noindent {\bf b.} (Mean functions) For  $r >0$ and $L>0$, we have $  f^*_t \in \mathcal{H}^r (L) $.  The minimal spacing between two consecutive change-points $\Delta = \min_{k = 1}^{ K+1}(\eta_{k} - \eta_{k-1})$ satisfies that $\Delta = \Theta (T)$.  

\noindent {\bf c.} (Functional noise)  Let $\{\varepsilon_i, \varepsilon'_0\}_{i\in \mathbb Z}$ be i.i.d.~random elements taking values in a measurable space $S_\xi$ and $g$ be a measurable function $g:\, S_\xi^\infty \to \mathcal{L}_2$. The functional noise $\{\xi_t\}_{t=1}^T  \subseteq \lt$ takes the form
\[
    \xi_t =g ( \mathcal G_t  ), \quad \mbox{with } \mathcal G_t   =( \ldots, \varepsilon_{-1},   \varepsilon_0, \varepsilon_1, \ldots,  \varepsilon_{t-1} , \varepsilon_t).
\]
There exists an absolute constant $q \geq 3$, such that $\mathbb{E} ( \| \xi_t\|_\infty ^q ) < C_{\xi, 1}$ for some absolute constant $C_{\xi, 1}$.  Define a coupled process
\[
    \xi^*_t =g ( \mathcal G^*_t  ), \quad \mbox{with } \mathcal G^*_t   =( \ldots, \varepsilon_{-1},   \varepsilon'_0, \varepsilon_1, \ldots,  \varepsilon_{t-1} , \varepsilon_t).
\]
We have $\sum_{t=1}^\infty t^{1/2-1/q}  \{   \mathbb{E}\| \xi_t -\xi_t^*\|_{\infty }^q  \} ^{1/q} < C_{\xi, 2}$  for some absolute constant $C_{\xi, 2} > 0$.

\noindent {\bf d.} (Measurement error) Let $\{\epsilon_i, \epsilon'_0\}_{i \in \mathbb{Z}}$ be i.i.d.~random elements taking values in a measurable space $S_\delta$ and $\tilde{g}_n$ be a measurable function $\tilde{g}_n:\, S_\delta^\infty \to \mathbb{R}^n$.  The measurement error $\{\delta_t\}_{t = 1}^T \subseteq \mathbb{R}^n$ takes the form
\[
    \delta_t = \tilde{g}_n ( \mathcal F _t ), \quad \mbox{with } \mathcal F_t   =( \ldots, \epsilon_{-1},   \epsilon_0, \epsilon_1, \ldots,  \epsilon_{t-1} , \epsilon_t).
\]
There exists an absolute constant $q \geq 3$, such that $\max_{ i = 1}^{ n} \mathbb{E}(| \delta_{t,i}|^q) < C_{\delta, 1}$ for some absolute constant $C_{\delta, 1}$.  Define a coupled process
\[
    \delta_t^* = \tilde{g}_n ( \mathcal F_t^* ), \quad \mbox{with } \mathcal F_t^*   =( \ldots, \epsilon_{-1},   \epsilon_0', \epsilon_1, \ldots,  \epsilon_{t-1} , \epsilon_t).
\]
We have $\max_{i = 1}^{ n}  \sum_{t=1}^\infty t^{1/2-1/q } \{\mathbb{E} | \delta _{t,i} -\delta_{t,i} ^* | ^q  \} ^{1/q} < C_{\delta, 2}$
for some absolute constant $C_{\delta, 2} > 0$.
\end{assumption}

\Cref{assume: model assumption}\textbf{a} allows the functional data to be observed on discrete grids and moreover, we allow for different grids at different time points. The sampling distribution $\mu$ is required to be lower bounded on the support $[0, 1]^d$, which is a standard assumption widely used in the nonparametric literature \citep[e.g.][]{Tsybakov:1315296}. Here, different functional curves are assumed to have the same number of grid points $n$. We remark that this is made for presentation simplicity only. It can indeed be further relaxed and the main results below will then depend on both the minimum and maximum numbers of grid points.

Note that \Cref{assume: model assumption}\textbf{a} does not impose any restriction between the sampling frequency $n$ and the number of functional curves $T$, and indeed our method can handle both the dense case where $n\gg T$ and the sparse case where $n$ can be upper bounded by a constant. Besides the random sampling scheme studied here, another commonly studied scenario is the fixed design, where it   usually assumes that the sampling locations $\{x_i \}_{i=1}^n$ are common to all functional curves across time. We remark that while we  focus  on the random design here, our  proposed algorithm can be directly applied to the fixed design case without any modification. Furthermore, its theoretical justification under the fixed design case can be established similarly with minor modifications, which is omitted.


The observed functional data have mean functions $\{f^*_t\}_{t = 1}^T$, which are assumed to be H\"{o}lder continuous in \Cref{assume: model assumption}\textbf{b}.  Note that the H\"{o}lder parameters in \Cref{assume: model assumption}\textbf{a} and \textbf{b} are both denoted by $r$.  We remark that different smoothness are allowed and we use the same $r$ here for notational simplicity.  This sequence of mean functions is our primary interest and is assumed to possess a piecewise constant pattern, with the minimal spacing $\Delta$ being of the same order as $T$. This assumption essentially requires that the number of change-points is upper bounded. It can also be further relaxed and we will have more elaborated discussions on this matter in \Cref{sec-conclusion}.

Our model allows for two sources of noise - functional noise and measurement error, which are detailed in \Cref{assume: model assumption}\textbf{c} and \textbf{d}, respectively. Both the functional noise and the measurement error are allowed to possess temporal dependence and heavy-tailedness.  For temporal dependence, we adopt the physical dependence framework by \cite{wu2005nonlinear}, which covers a wide range of time series models, such as ARMA and vector AR models.  It further covers popular functional time series models such as functional AR and MA models~\citep{hormann2010weakly}.  We also remark that \Cref{assume: model assumption}\textbf{c}~and \textbf{d} impose a short range dependence, which is characterized by the absolute upper bounds $C_{\xi,2}$ and $C_{\delta,2}$. Further relaxation is possible by allowing the upper bounds $C_{\xi,2}$ and $C_{\delta,2}$ to vary with the sample size $T$. 

The heavy-tail behavior is encoded in the parameter $q$.  In \Cref{assume: model assumption}\textbf{c} and \textbf{d}, we adopt the same quantity $q$ for presentational simplicity and remark that different heavy-tailedness levels are allowed.  An extreme example is that when $q = \infty$, the noise is essentially sub-Gaussian.  Importantly, \Cref{assume: model assumption}\textbf{d} does not impose any restriction on the cross-sectional dependence among measurement errors observed on the same time $t$, which can be even perfectly correlated.

\subsection{Kernel-based change-point detection}
To estimate the change-point $\{\eta_k\}_{k=1}^K$ in the mean functions $\{f^*_t\}_{t = 1}^T$, we propose a kernel-based cumulative sum~(CUSUM) statistic, which is simple, intuitive and computationally efficient. The key idea is to recover the unobserved $\{f^*_t\}_{t = 1}^T$ from the observations $\{(x_{t, i}, y_{t, i})\}_{t = 1,i = 1}^{T,n}$ based on kernel estimation.

Given a kernel function $K(\cdot): \mathbb R^d \to \mathbb R^+$ and a bandwidth parameter $h > 0$, we define $K_h(x) = h^{-d} K(x/h)$ for $x \in \mathbb R^d.$ Given the random grids $\{x_{t,i}\}_{t=1,i=1}^{T,n}$ and a bandwidth parameter $\bar{h}$, we define the density estimator of the sampling distribution $u(x)$ as
\[
    \hat{p}(x) = \hat{p}_{\bar{h}}(x) =\frac{1}{nT}\sum_{t=1}^{T}\sum_{i=1}^{n}K_{\bar{h}}(x-x_{t,i}), \quad x \in [0, 1]^d.
\]

Given $\hat p(x)$ and a bandwidth parameter $h>0$, for any time $t=1,2,\cdots, T$, we define the kernel-based estimation for $f_t^*(x)$ as
\begin{equation}\label{eq-F-t-h-x-def}
    F_{t, h} (x) = \frac{\sum_{i=1}^n y_{t, i} K_h(x-x_{t, i }) }{n \hat{p}(x)}, \quad x \in [0, 1]^d.
\end{equation}
Based on the kernel estimation $F_{t,h}(x)$, for any integer pair $0 \leq s < e \leq T$, we define the CUSUM statistic as 
\begin{equation}\label{def-cusum}
    \widetilde{F}_{t, h}^{(s, e]}(x) = \sqrt{\frac{e-t}{(e-s )(t-s)}} \sum_{ l =s+1}^{ t} F_{ l , h  } (x) - \sqrt { \frac{t-s}{ (e-s )(e-t)}} \sum_{l =t+1}^{ e} F_{l , h   } (x), \quad x \in [0, 1]^d.
\end{equation}

The CUSUM statistic defined in \eqref{def-cusum} is the cornerstone of our algorithm and is based on two kernel estimators $\hat{p}(\cdot)$ and $F_{t, h}(\cdot)$. At a high level, the CUSUM statistic  $\widetilde{F}_{t, h}^{(s, e]}(\cdot ) $   estimates the  difference in mean between the functional data in the time intervals $(s,t] $ and $(t,e] $.
In the functional data analysis literature, other popular approaches  for mean function estimation are reproducing kernel Hilbert space based methods and local polynomial regression. However, to our best knowledge, existing works based on the two approaches typically require that the functional data are temporally independent and it is not obvious how to extend their theoretical guarantees to the temporal dependence case. We therefore choose the kernel estimation method owing to its flexibility in terms of both methodology and theory and we derive new theoretical results on kernel estimation for functional data under temporal dependence and heavy-tailedness.

For multiple change-point estimation, a key ingredient is to isolate each single change-point with well-designed intervals in $[0,T]$. To achieve this, we combine the CUSUM statistic in \eqref{def-cusum} with a modified version of the seeded binary segmentation~(SBS) proposed in \cite{kovacs2020seeded}. SBS is based on a collection of deterministic intervals defined in \Cref{definition:seeded}.

\begin{definition} [Seeded intervals]\label{definition:seeded}
Let $\mathcal{K} = \lc C_{\mathcal{K}}\loglog(T) \rc$, with some sufficiently large absolute constant $C_{\mathcal{K}} > 0$.  For $k \in \{1, \ldots, \mathcal{K}\}$, let $\mathcal{J}_k$ be the collection of $2^k - 1$ intervals of length $l_k = T2^{-k+1}$ that are evenly shifted by $l_k/2 = T2^{-k}$, i.e.
$$\mathcal{J}_k = \{(\lfloor (i -1) T2^{-k} \rfloor, \, \lc (i-1) T2^{-k} + T2^{-k+1}\rc ], \quad i = 1, \ldots,  2^k - 1 \}.$$
The overall collection of seeded intervals is denoted as $\mathcal{J} = \cup_{k = 1}^{\mathcal{K}} \mathcal{J}_k$.
\end{definition}

The essential idea of the seeded intervals defined in \Cref{definition:seeded} is to provide a multi-scale system of searching regions for multiple change-points. SBS is computationally efficient with a computational cost of the order $O\{T\log(T)\}$~\citep{kovacs2020seeded}.

Based on the CUSUM statistic and seeded intervals, \Cref{alg:FSBS} summarises the proposed functional seeded binary segmentation algorithm~(FSBS) for multiple change-point estimation in sequentially observed functional data. There are three main tuning parameters involved in \Cref{alg:FSBS}, the kernel bandwidth $\bar{h}$ in the estimation of the sampling distribution, the kernel bandwidth $h$ in the estimation of the mean function and the threshold parameter $\tau$ for declaring change-points. Their theoretical and numerical guidance will be presented in Sections~\ref{sec-theory} and \ref{sec-numeric}, respectively.

\begin{algorithm}[ht]
\begin{algorithmic}
    \INPUT  Data $\{x_{t, i}, y_{t, i}\}_{t = 1, i = 1}^{T,n}$, seeded intervals $\mathcal{J}$, tuning parameters $\bar{h}, h, \tau>0$.
    \State \textbf{Initialization}: If $(s,e]=(0,n]$, set $\textbf{S} \leftarrow \varnothing$ and set $\rho \leftarrow \log(T)n^{-1}h^{-d}$. Furthermore, sample $\lceil \log(T) \rceil$ points from $\{x_{t, i}\}_{t = 1, i = 1}^{T,n}$ uniformly at random without replacement and denote them as $\{u_m\}_{m=1}^{\lceil \log(T) \rceil}$. Estimate the sampling distribution evaluated at $\{\hat{p}_{\bar{h}}(u_m)\}_{m=1}^{\lceil \log(T) \rceil}$. 
    
    \bigskip
    \For {$\mathcal I=(\alpha, \beta ] \in \mathcal{J}$ and $m \in \{1, \ldots, \lceil \log(T) \rceil$\}}
		\If {$ \mathcal I= (\alpha, \beta ] \subseteq  (s,e]  $ and $\beta-\alpha >2\rho $} 
            \State $A_m^\mathcal I    \leftarrow \max_{\alpha +\rho       \leq t \leq  \beta -\rho  }  | \widetilde F_{t, h }^{ (\alpha  , \beta ]   } (u_m ) | $ 
            \State $D _m^\mathcal I  \leftarrow \argmax_{\alpha+\rho   \leq t \leq  \beta-\rho   }  | \widetilde F_{t, h    }^{ (\alpha  , \beta ]   } (u_m ) | $ 
        \Else
            \State{$(A_m^{\mathcal{I}}, D_m^{\mathcal{I}}) \leftarrow (-1, 0)$} 
        \EndIf  
   \EndFor 
    \State $(m^*, \mathcal I^* )  \leftarrow \argmax_{m = 1, \ldots, \lceil \log(T)\rceil, \mathcal{I} \in \mathcal{J}} A^\mathcal I _m $.
    \If{$A_{m^*} ^{\mathcal I^*} >\tau$}
        \State ${\bf S} \leftarrow {\bf S} \cup D_{m^*} ^{\mathcal I^*}$ 
        \State  FSBS $( (s,D_{m^*} ^{\mathcal I^*}], \bar{h}, h    ,\tau )$   
 		\State  FSBS $( ( D_{m^*} ^{\mathcal I^*},e ], \bar{h}, h    , \tau   )$   
    \EndIf 
    \bigskip
    
    \OUTPUT The set of estimated change-points $\textbf{S}$.
    \caption{Functional Seeded Binary Segmentation.  FSBS $( (s,e], \bar{h}, h ,   \tau   )$} \label{alg:FSBS}
\end{algorithmic}
\end{algorithm}

\Cref{alg:FSBS} is conducted in an iterative way, starting with the whole time course, using the multi-scale seeded intervals to search for the point according to the largest CUSUM value.  A change-point is declared if the corresponding maximum CUSUM value exceeds a pre-specified threshold $\tau$ and the whole sequence is then split into two with the procedure being carried on in the sub-intervals.


\Cref{alg:FSBS} utilizes a collection of  random grid points  $\{u_m\}_{m=1}^{\lceil \log(T) \rceil} \subseteq  \{x_{t, i}\}_{t = 1, i = 1}^{T,n} $ to detect changes in the functional data. For a change of mean functions at the time point $\eta$ with $\|f^*_{\eta +1} - f^*_{\eta }\|_\infty> 0$, we show in the appendix that, as long as $\lceil \log(T) \rceil$ grid points are sampled, with high probability, there is at least one point $ u_{m'} \in \{u_m\}_{m=1}^{\lceil \log(T) \rceil} $ such that  $  |f^*_{\eta +1} (u_{m'}) - f^*_{\eta }(u_{m'})  | \asymp  \|f^*_{\eta +1} - f^*_{\eta }\|_\infty. $ Thus, this procedure allows FSBS to detect changes in the  mean functions without evaluating functions  on a  dense lattice grid and thus improves computational efficiency.  


\section{Main Results} \label{sec-main}
 
\subsection{Assumptions and theory}\label{sec-theory}

We begin by imposing assumptions on the kernel function $K(\cdot)$ used in FSBS.
\begin{assumption}[Kernel function]\label{Kernel-as}
Let $K(\cdot) :\mathbb R^d \to \mathbb R^+  $   be  compactly supported and satisfy the following conditions.

\noindent {\bf a.} The kernel function $K(\cdot)$ is adaptive to the H\"older class $\hr(L)$, i.e.~for any $f \in \hr(L)$, it holds that $\sup_{x \in [0,1]^d} \big|\int_{[0,1]^d}  K_h\left({x-z}\right) f(z) \dint z - f(x)\big| \leq \tilde{C} h^r,$ where $ \tilde{C} > 0$ is a constant that only depends on $L$.

\noindent{\bf b.} The class of functions 
$\mathcal{F}_K = \{K(x-\cdot)/h:\,   \mathbb{R}^d\to \mathbb R^+, h > 0\}$ 
 is separable  in $\mathcal{L}_{\infty}(\mathbb{R}^d)$ and is a  uniformly bounded VC-class.  This means that there exist constants $A, \nu > 0$  such that for every probability measure $Q$ on $\mathbb{R}^d$ and every $u \in (0, \|K\|_{\infty})$, it holds that $\mathcal{N}(\mathcal{F}_K,\mathcal L_2(Q),u) \leq \left({A\|K\|_\infty}/{u}\right)^{v}$, where $\mathcal{N}(\mathcal{F}_K, \mathcal L_2(Q) ,u)$ denotes the $u$-covering number of the metric space $(\mathcal{F}_K, \mathcal L_2(Q))$.
\end{assumption} 

\Cref{Kernel-as} is a standard assumption in the nonparametric literature, see \cite{gine1999laws,gine2001consistency}, \cite{kim2019uniform}, \cite{wang2019dbscan} among many others. These assumptions hold for most commonly used kernels, including uniform, polynomial  and Gaussian kernels.

Recall the minimal spacing $\Delta = \min_{k = 1}^{ K+1}(\eta_{k} - \eta_{k-1})$ defined in \Cref{assume: model assumption}\textbf{b}. We further define the jump size at the $k$th change-point as $\kappa_k = \|f^*_{\eta_k+1} - f^*_{\eta_k}\|_\infty$ and define $\kappa=\min_{k = 1}^{ K}$ as the minimal jump size. \Cref{assume-snr} below details how strong the signal needs to be in terms of $\kappa$ and $\Delta$, given the grid size $n$, the number of functional curves $T$, smoothness parameter $r$, dimensionality $d$ and moment condition $q$. 

\begin{assumption}[Signal-to-noise ratio, SNR]\label{assume-snr}
There exists an arbitrarily-slow diverging sequence $C_{\mathrm{SNR}} =C_{\mathrm{SNR}}(T)$ such that
\[
    \kappa \sqrt{\Delta} > C_{\mathrm{SNR}} \log^{ \max\{1/2, 5/q\}}(T)\Big(1 + T^{\frac{d}{2r+d}} n^{\frac{-2r}{2r+d}} \Big)^{1/2}.
\]
\end{assumption}


We are now ready to present the main theorem, showing the consistency of FSBS.
\begin{theorem} \label{theorem:FSBS}
Under Assumptions~\ref{assume: model assumption}, \ref{Kernel-as} and \ref{assume-snr}, let $\{\widehat \eta_k\}_{k=1}^{\widehat K}$ be the estimated change-points by FSBS detailed in \Cref{alg:FSBS} with data $\{x_{t, i}, y_{t, i}\}_{t = 1,i=1}^{T,n}$, bandwidth parameters $\bar{h} = C_{\bar{h}} (Tn)^{-\frac{1}{2r+d}}$, $h = C_h (Tn)^{\frac{-1}{2r+d }}$ and threshold parameter $\tau = C_\tau \log^{ \max\{1/2, 5/q\}}(T)\left(1 + T^{\frac{d}{2r+d}} n^{\frac{-2r}{2r+d}} \right)^{1/2}$, 
for some absolute constants $C_{\bar{h}}, C_h, C_{\tau} > 0$. It holds that 
\begin{align*}
      \mathbb{P}\left\{\widehat K = K;\, |\widehat \eta_k - \eta_k| \leq C_{\mathrm{FSBS}} \log^{ \max\{1, 10/q \} }  (T)  \left(1 +   T^{\frac{d}{2r+d}}   n^{\frac{2r}{2r+d}}  \right) \kappa_k^{-2}, \, \forall k = 1, \ldots, K\right\}  
       \geq 1 - 3\log^{-1}(T), 
\end{align*}	
where $C_{\mathrm{FSBS}} > 0$ is an absolute constant.
\end{theorem}

In view of \Cref{assume-snr} and \Cref{theorem:FSBS}, we see that with properly chosen tuning parameters and with probability tending to one as the sample size $T$ grows, the output of FSBS estimates the correct number of change-points and 
\[
    \max_{k = 1 }^K  {|\widehat{\eta}_k - \eta_k|}/{\Delta} \lesssim {\left(1 + T^{\frac{d}{2r+d}} n^{\frac{-2r}{2r+d}} \right) \log^{ \max\{1, 10/q \} }  (T)}/{(\kappa^2 \Delta) } =  o(1),
\]
where the last inequality follows from \Cref{assume-snr}.  The above inequality shows that there exists a one-to-one mapping from $\{\widehat{\eta}_k\}_{k = 1}^K$ to $\{\eta_k\}_{k = 1}^K$, assigning by the smallest distance.

\subsection{Discussions on functional seeded binary segmentation~(FSBS)}\label{sec-comparison}

\textbf{From sparse to dense regimes.}  In our setup, each curve is only observed at $n$ discrete points and we allow the full range of choices of $n$, representing from sparse to dense scenarios, all accompanied with consistency results.  In the most sparse case $n = 1$, \Cref{assume-snr} reads as $\kappa \sqrt{\Delta} \gtrsim T^{d/(4r+2d)} \times$ a logarithmic factor, under which the localisation error is upper bounded by $T^{d/(2r+d)} \kappa^{-2}$, up to a logarithmic factor.  To the best of our knowledge,   this  challenging case  has not been  dealt in the existing change-point detection literature for functional data.  In the most dense case, we can heuristically let $n = \infty$ and for simplicity let $q = \infty$ representing the sub-Gaussian noise case.  \Cref{assume-snr} reads as $\kappa\sqrt{\Delta} \asymp \log^{1/2}(T)$ and the localisation error is upper bounded by $\kappa^{-2}\log(T)$.  Both the SNR ratio and localisation error are the optimal rate in the univariate mean change-point localisation problem \citep{wang2020univariate}, implying the optimality of FSBS in the dense situation.

\textbf{Tuning parameters.}  There are three tuning parameters involved.  In the CUSUM statistic \eqref{def-cusum}, we specify that the density estimator of the sampling distribution is a kernel estimator with bandwidth $\bar{h} \asymp (Tn)^{-1/(2r+d)}$.  Due to the independence of the observation grids, such a choice of the bandwidth follows from the classical nonparametric literature \citep[e.g.][]{Tsybakov:1315296} and is minimax-rate optimal in terms of the estimation error. For completeness, we include the study of $\hat{p}(\cdot)$'s theoretical properties in \Cref{sec-proof-thm1}.  In practice, there exist  different default methods for the selection of $\bar{h}$ , see for example the function \texttt{Hpi} from the R package \texttt{ks}  (\cite{chacon2018multivariate}).


The other bandwidth tuning parameter $h$ is also required to be $h \asymp ( Tn)^{-1/(2r+d)}$.  Despite that we allow for physical dependence in both functional noise and measurement error, we show that the same order of bandwidth (as $\bar{h}$) is required under \Cref{assume: model assumption}.  This is an interesting finding, if not surprising.  This particular choice of $h$ is due to the fact that the physical dependence put forward by \cite{wu2005nonlinear} is a short range dependence condition and does not change the rate of the sample size.

The threshold tuning parameter $\tau$ is set to be a high-probability upper bound on the CUSUM statistics when there is no change-point  and is in fact of the form
\[
    \tau = C_\tau \log^{ \max\{1/2, 5/q\}}(T) \sqrt {n^{-1}h^{-d} + 1}.
\]
 This also reflects the requirement on the SNR detailed in \Cref{assume-snr}, that $\kappa\sqrt{\Delta} \gtrsim \tau$.

\textbf{Phase transition.}  Recall that the number of curves is $T$ and the number of observations on each curve is $n$.  The asymptotic regime we discuss is to let $T$ diverge, while allowing all other parameters, including~$n$, to be functions of $T$.  In \Cref{theorem:FSBS}, we allow a full range of cases in terms of the relationship between $n$ and $T$.  As a concrete example, when the smooth parameter $r = 2$, the jump size $\kappa \asymp 1$ and in the one-dimensional case $d = 1$, with high probability (ignoring logarithmic factors for simplicity),
\[
    \max_{k = 1}^{K} |\widehat{\eta}_k - \eta_k| =O_p( T^{\frac{1}{5}} n^{-\frac{4}{5}} + 1 ) = \begin{cases}
        O_p(1), & n \geq T^{1/4}; \\ 
        O_p(T^{\frac{1}{5}} n^{-\frac{4}{5}}), &  n \leq T^{1/4}.
    \end{cases} 
\]

This relationship between $n$ and $T$ was previously demonstrated in the mean function estimation literature \citep[e.g.][]{cai2011optimal, zhang2016sparse}, where the observations are discretely sampled from independently and identically distributed functional data.
It is shown that the minimax estimation error rate also possesses the same phase transition between $n$ and $T$, i.e.~with the transition boundary $n \asymp T^{1/4}$, which agrees with our finding under the change-point setting. 

\textbf{Physical dependence and heavy-tailedness}  In \Cref{assume: model assumption}{\bf c} and {\bf d}, we allow for physical dependence type temporal dependence and heavy-tailed additive noise.  As we have discussed, since the physical dependence is in fact a short range dependence, all the rates involved are the same as those in the independence cases, up to logarithmic factors.  Having said this, the technical details required in dealing with this short range dependence are fundamentally different from those in the independence cases.  From the result, it might be more interesting to discuss the effect of the heavy-tail behaviours, which are characterised by the parameter $q$.  It can be seen from the rates in \Cref{assume-snr} and \Cref{theorem:FSBS} that the effect of $q$ disappears and it behaves the same as if the noise is sub-Gaussian when $q \geq 10$.
 



\section{Numerical Experiments}\label{sec-numeric}

\subsection{Simulated data analysis}\label{simu-data}
We compare the proposed FSBS with state-of-the-art methods for change-point detection in functional data across a wide range of simulation settings. The implementations for our approaches can be found at  \href{https://github.com/cmadridp/FSBS}{https://github.com/cmadridp/FSBS}.  We compare with three competitors: BGHK in \cite{berkes2009detecting}, HK in \cite{hormann2010weakly} and SN in \cite{zhang2011testing}. All three methods estimate change-points via examining mean change in the leading functional principal components of the observed functional data. BGHK is designed for temporally independent data while HK and SN can handle temporal dependence via the estimation of long-run variance and the use of self-normalization principle, respectively. All three methods require fully observed functional data. In practice, they convert discrete data to functional observations by using B-splines with 20 basis functions.


For the implementation of FSBS, we adopt the Gaussian kernel. Following the standard practice in kernel density estimation, the bandwidth $\bar{h}$ is selected by the function \texttt{Hpi} in the R package \texttt{ks}~(\cite{chacon2018multivariate}).  The tuning parameter $\tau$ and the bandwidth $h$ are chosen by cross-validation, with evenly-indexed data being the training set and oddly-indexed data being the validation set.  For each pair of candidate $(h, \tau)$, we obtain change-point estimators $\{\widehat{\eta}_k\}_{k = 1}^{\widehat{K}}$ on the training set and compute the validation loss $\sum_{k = 1}^{\widehat{K}} \sum_{t \in [\widehat{\eta}_k, \widehat{\eta}_{k+1})} \sum_{i=1}^n\{(\widehat{\eta}_{k+1} - \widehat{\eta}_k)^{-1}\sum_{t=\widehat{\eta}_{k}+1}^{\widehat{\eta}_{k+1}}F_{t,h}(x_{t,i})-y_{t,i}\}^2.$ The pair $(h, \tau)$ is then chosen to be the one corresponding to the lowest validation loss.

%
%



We consider five different scenarios for the observations $\{x_{ti},y_{ti}\}_{t=1,i=1}^{T,n}$. For all scenarios 1-5, we set $T=200.$ Given the dimensionality $d$, denote a generic grid point as $x=(x^{(1)},\cdots,x^{(d)})$. Scenarios 1 to 4 are generated based on model \eqref{eq_model}. The basic setting is as follows.

$\bullet$ {\bf{Scenario 1} (S1)}  Let $(n, d) = (1, 1)$, the unevenly-spaced change-points be $(\eta_1, \eta_2) = (30, 130)$ and the three distinct mean functions be $6\cos(\cdot)$, $6\sin(\cdot)$ and $6\cos(\cdot)$.

$\bullet$  {\bf{Scenario 2} (S2)} Let $(n, d) = (10, 1)$, the unevenly-spaced change-points be $(\eta_1, \eta_2) = (30, 130)$ and the three distinct mean functions be $2\cos(\cdot)$, $2\sin(\cdot)$ and $2\cos(\cdot)$.

$\bullet$ {\bf{Scenario 3} (S3)}  Let $(n, d) = (50, 1)$, the unevenly-spaced change-points be $(\eta_1, \eta_2) = (30, 130)$ and the three distinct mean functions be $\cos(\cdot)$, $\sin(\cdot)$ and $\cos(\cdot)$.
    
$\bullet$  {\bf{Scenario 4} (S4)} Let $(n, d) = (10, 2)$, the unevenly-spaced change-points be $(\eta_1, \eta_2) = (100, 150)$ and the three distinct mean functions be $0$, $3x^{(1)}x^{(2)}$ and $0$. 

For \textbf{S1}-\textbf{S4}, the functional noise is generated as $\xi_t (x)=0.5\xi_{t-1} (x)+\sum_{i=1}^{50} i^{-1} b_{t,i} h_{i}(x)$, where $\{h_i(x)=\prod_{j=1}^{d}(1/\sqrt{2})\pi\sin(ix^{(j)})\}_{i=1}^{50}$ are basis functions and $\{b_{t,i}\}_{t = 1, i=1}^{T, 50}$ are i.i.d.~standard normal random variables.  The measurement error is generated as $\delta_{t}=0.3\delta_{t-1}+\epsilon_t$, where $\{\epsilon_t\}_{t = 1}^T$ are i.i.d.~$\mathcal{N}(0,0.5I_n)$. We observe the noisy functional data $\{y_{ti}\}_{t=1,i=1}^{T,n}$ at grid points $\{x_{ti}\}_{t=1,i=1}^{T,n}$ independently sampled from $\mathrm{Unif}([0,1]^d)$.





Scenario 5 is adopted from \cite{zhang2011testing} for densely-sampled functional data without measurement error.

$\bullet$  {\bf{Scenario 5} (S5)} Let $(n, d) = (50, 1)$, the evenly-spaced change-points be $(\eta_1, \eta_2) = (68, 134)$ and the three distinct mean functions be 0, $\sin(\cdot)$ and $2\sin(\cdot)$.

The grid points $\{x_{ti}\}_{i=1}^{50}$ are $50$ evenly-spaced points in $[0, 1]$ for all $t=1,\cdots, T$. The functional noise is generated as
$\xi_t(\cdot) = \int_{[0,1]}\psi(\cdot, u)\xi_{t-1}(u) \dint u +\epsilon_t(\cdot)$, where $\{\epsilon_t(\cdot)\}_{t = 1}^T$ are independent standard Brownian motions and $\psi(v, u)=1/3\exp((v^2+u^2)/2)$ is a bivariate Gaussian kernel. 


\textbf{S1-S5} represent a wide range of simulation settings including the extreme sparse case \textbf{S1}, sparse case \textbf{S2}, the two-dimensional domain case \textbf{S4}, and the densely sampled cases \textbf{S3} and \textbf{S5}. Note that \textbf{S1} and \textbf{S4} can only be handled by FSBS as for \textbf{S1} it is impossible to estimate a function via B-spline based on one point and for \textbf{S4}, the domain is of dimension 2.

\textbf{Evaluation result}: For a given set of true change-points $\mathcal{C}= \{\eta_k\}_{k = 1}^K$, we evaluate the accuracy of the estimator $\{\widehat{\eta}_k\}_{k = 1}^{\widehat{K}}$ by the difference $|\widehat{K} - K|$ and the Hausdorff distance $d(\hat{\mathcal{C}},\mathcal{C})$, defined by $d(\hat{\mathcal{C}},\mathcal{C})=\max\{\max_{x\in{\hat{\mathcal{C}}}}\min_{y\in{\mathcal{C}}}\{\vert x-y\vert\},\max_{y\in{\hat{\mathcal{C}}}}\min_{x\in{\mathcal{C}}}\{\vert x-y\vert\}\}$. For $\hat{\mathcal{C}} = \varnothing$, we use the convention that  $|\widehat{K}-K|=K$ and $d(\hat{\mathcal{C}},\mathcal{C})=T$.

For each scenario, we repeat the experiments 100 times and \Cref{fig:1} summarizes the performance of FSBS, BGHK, HK and SN. Tabulated results can be found in \Cref{sec-tables}. As can be seen, FSBS consistently outperforms the competing methods by a wide margin and demonstrates robust behaviour across the board for both sparsely and densely sampled functional data.

\begin{figure}[h]
    \centering
    \includegraphics[width=\textwidth]{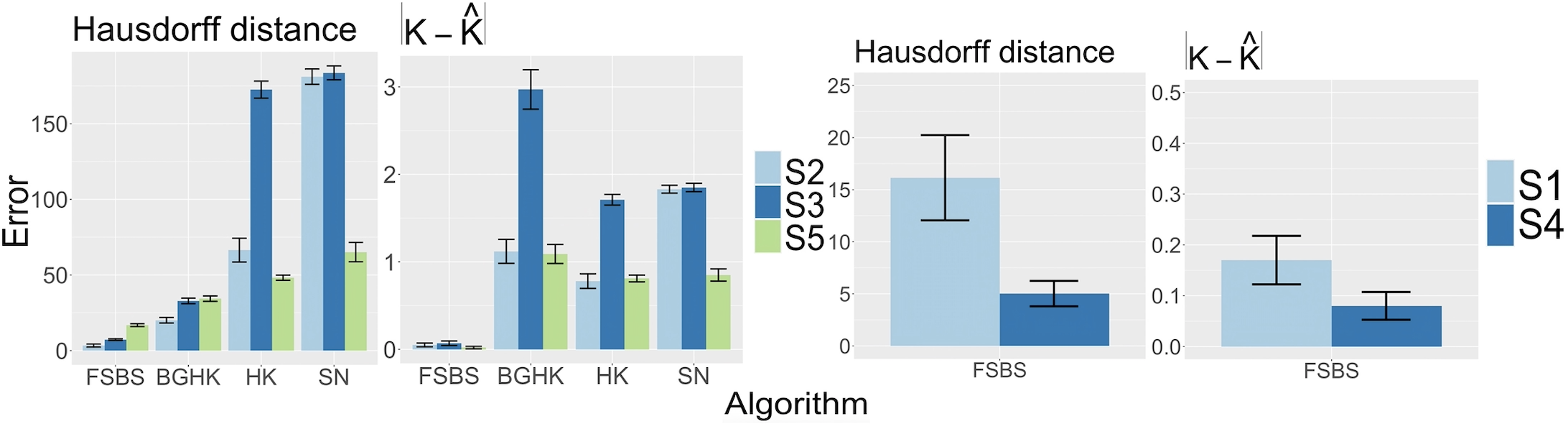}
    
    \caption{Bar plots for simulation results of {\bf{S1-S5}}. Each bar reports the mean and standard deviation computed based on 100 experiments. From left to right, the first two plots correspond to the Hausdorff distance and $\vert K-\hat{K}\vert$ in {\bf{S2}}, {\bf{S3}} and {\bf{S5}}. The last two plots correspond to {\bf{S1}} and {\bf{S4}}. } 
    \label{fig:1}

\end{figure}

\subsection{Real data application } \label{real-data}
We consider the COBE-SSTE dataset \citep{dataset}, which consists of monthly average sea surface temperature (SST) from 1940 to 2019, on a $1$ degree latitude by $1$ degree longitude grid $(48 \times 30)$ covering Australia.  The specific coordinates are latitude $10$S-$39$S and longitude $110$E-$157$E. 

We apply FSBS to detect potential change-points in the two-dimensional SST. The implementation of FSBS is the same as the one described in \Cref{simu-data}. To avoid seasonality, we apply FSBS to the SST for the month of June from 1940 to 2019. We further conduct the same analysis separately for the month of July for robustness check.

For both the June and July data, two change-points are identified by FSBS, Year 1981 and 1996, suggesting the robustness of the finding. The two change-points might be associated with years when both the Indian Ocean Dipole and Oceanic Ni\~{n}o Index had extreme events~\citep{2003}. The El Ni\~{n}o/Southern Oscillation has been recognized as an important manifestation of the tropical ocean-atmosphere-land coupled system. It is an irregular periodic variation in winds and sea surface temperatures over the tropical eastern Pacific Ocean. Much of the variability in the climate of Australia is connected with this phenomenon \citep{link1}. 

To visualize the estimated change, \Cref{fig:2} depicts the average SST before the first change-point Year 1981, between the two change-points, and after the second change-point Year 1996.  The two rows correspond to the June and July data, respectively.  As we can see, the top left corners exhibit different patterns in the three periods, suggesting the existence of change-points.

\begin{figure}[h]
    \centering
    \includegraphics[width=0.7\textwidth]{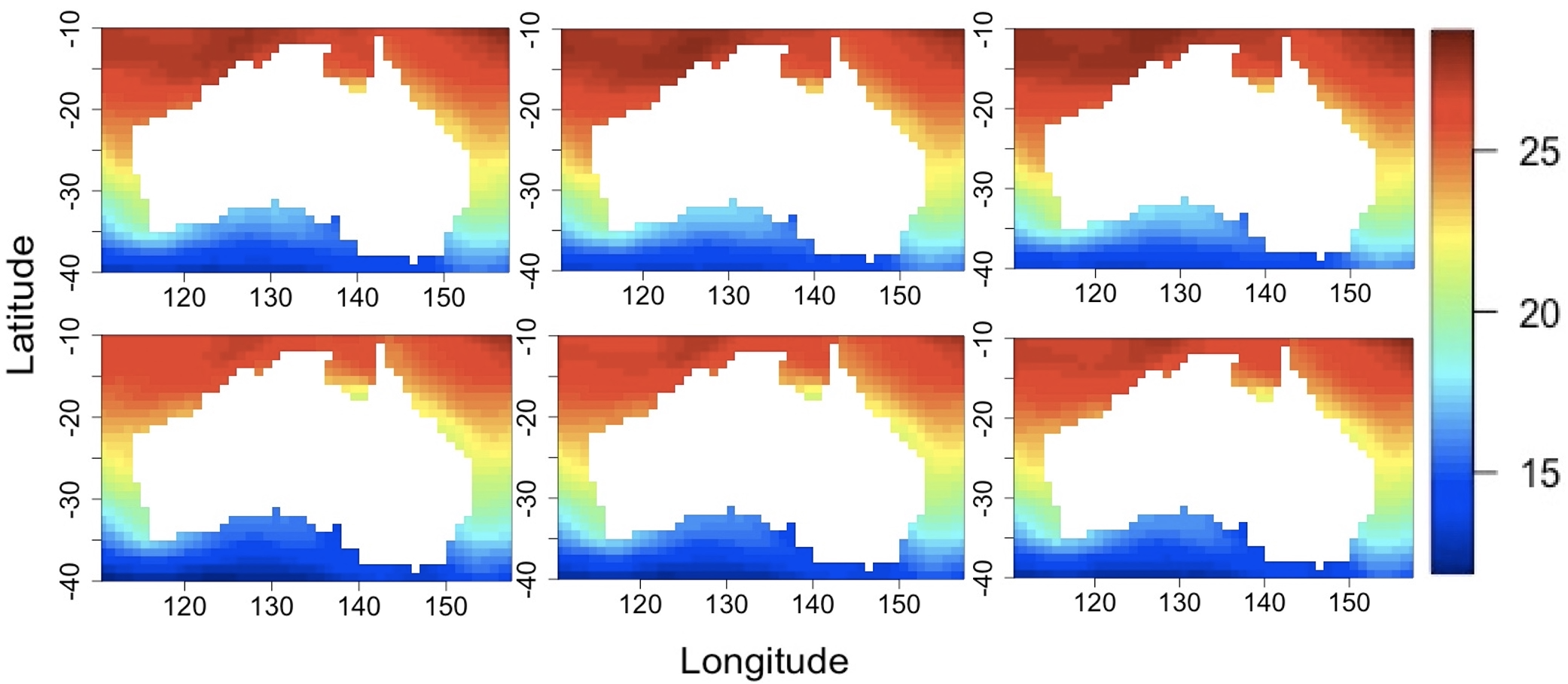}
    \caption{Average SST.  From left to right: average SST from 1940 to 1981, average SST from 1982 to 1996, and average SST from 1997 to 2019.  The top and bottom rows correspond to the June and July data respectively.}
    \label{fig:2}
\end{figure}

\section{Conclusion}\label{sec-conclusion}
In this paper, we study change-point detection for sparse and dense functional data in general dimensions. We show that our algorithm FSBS can consistently estimate the change-points even in the extreme sparse setting with $n=1$. Our theoretical analysis reveals an interesting phase transition between $n$ and $T$, which has not been discovered in the existing literature for functional change-point detection. The consistency of FSBS relies on the assumption that the minimal spacing $\Delta\asymp T$. To relax this assumption, we may consider increasing $\mathcal{K} $ in \Cref{definition:seeded} to enlarge  the coverage of the  seeded intervals in FSBS and apply the narrowest over threshold selection method \citep[Theorem 3 in][]{kovacs2020seeded}.  With minor modifications of the current theoretical analysis, the consistency of FSBS can be established for the case of  $ \Delta \ll T$.  Since such a relaxation  does not add much more methodological insights to our paper, we omit this  additional  technical discussion for conciseness. 

\newpage
\bibliographystyle{plainat}
\bibliography{References}
 
\clearpage
\appendix

\appendix
\section*{Appendices}
\section{Detailed simulation results}\label{sec-tables}
We present the tables containing the results of the simulation study in Section \ref{simu-data} of the main text. 
\begin{table}[H]
\centering
\caption{{\bf{Scenario 1}} ($n=1,d=1$ changes from $6\cos$-$6\sin$-$6\cos$)}
\medskip
\begin{tabular}{lccccc}
\hline
Model & $K- \hat{K}<0$ & $K- \hat{K}=0$& $K- \hat{K}>0$& $\vert \hat{K}-K\vert$ & $d$\\
\hline
FSBS  & 0.05 & 0.86& 0.09& $0.17$    & $16.15$ \\
\hline
\end{tabular}
\label{table:1}

\medskip
Changes occur at the times $30$ and $130$.
\end{table}%
%
\begin{table}[H]
\centering
\caption{{\bf{Scenario 2}} ($n=10,d=1$, changes from $2\cos$-$2\sin$-$2\cos$)}
\medskip
\begin{tabular}{lccccc}
\hline
Model & $K- \hat{K}<0$ & $K- \hat{K}=0$& $K- \hat{K}>0$& $\vert \hat{K}-K\vert$ & $d$\\
\hline
FSBS  & 0.05 & 0.95& 0& $0.05$    & $3.32$ \\
BGHK & $0.58$&$0.42$&$0$&$1.12$&$20.11$\\
HK & $0.16$&$0.47$&$0.37$&$0.78$&$66.45$\\
SN & $0.04$&$0.03$&$0.93$&$1.83$&$181.11$\\
\hline
\end{tabular}

\label{table:2}
\medskip
Changes occur at the times $30$ and $130$.
\end{table}%
%
%

\begin{table}[H]
\centering
\caption{{\bf{Scenario 3}} ($n=50,d=1$, changes from $\cos$-$\sin$-$\cos$)}
\medskip
\begin{tabular}{lccccc}
\hline
Model & $K- \hat{K}<0$ & $K- \hat{K}=0$& $K- \hat{K}>0$& $\vert \hat{K}-K\vert$ & $d$\\
\hline
FSBS  & 0 & 0.93& 0.07& $0.07$    & $7.35$ \\
BGHK & $0.85$&$0.15$&$0$&$2.97$&$32.88$\\
HK & $0$&$0.08$&$0.92$&$1.71$&$172.52$\\
SN & $0.02$&$0.04$&$0.94$&$1.85$&$183.63$\\
\hline
\end{tabular}

\label{table:3}
\medskip
Changes occur at the times $30$ and $130$.
\end{table}%
%
\begin{table}[H]
\centering
\caption{{\bf{Scenario 4}} ($n=10,d=2$, changes from $0$-$3x^{(1)}x^{(2)}$-$0$)}
\medskip
\begin{tabular}{lccccc}
\hline
Model & $K- \hat{K}<0$ & $K- \hat{K}=0$& $K- \hat{K}>0$& $\vert \hat{K}-K\vert$ & $d$\\
\hline
FSBS  & 0 & 0.92& 0.08& $0.08$    & $5.02$ \\
\hline
\end{tabular}
\label{table:4}
\medskip

Changes occur at the times $100$ and $150$.
\end{table}%
%

%

\begin{table}[H]
\centering
\caption{{\bf{Scenario 5}} ($n=50,d=1$, changes from $0$-$\sin$-$2\sin$)}
\medskip
\begin{tabular}{lccccc}
\hline
Model & $K- \hat{K}<0$ & $K- \hat{K}=0$& $K- \hat{K}>0$& $\vert \hat{K}-K\vert$ & $d$\\
\hline
FSBS  & 0.02 & 0.98& 0& $0.02$    & $16.9$ \\
BGHK & $0.48$&$0.30$&$0.22$&$1.09$&$34.36$\\
HK & $0$&$0.19$&$0.81$&$0.81$&$48.24$\\
SN & $0.08$&$0.33$&$0.59$&$0.85$&$65.15$\\
\hline
\end{tabular}

\label{table:5}
\medskip
Changes occur at the times $68$ and $134$.
\end{table}%
%
%

%
 
\section[]{Proof of \Cref{theorem:FSBS}}\label{sec-proof-thm1}

In this section, we present the proofs of theorem \Cref{theorem:FSBS}. To this end, we will invoke the following well-known $l_{\infty}$ bounds for kernel density estimation. 
\begin{lemma}\label{lemma:KDE}
Let $\{ x_{t,i}  \}_{i = 1, t = 1}^{n, T}$ be random grid points independently  sampled from a common density function $u: [0,1]^d \to \mathbb R$. Under Assumption~\ref{Kernel-as}-{\bf{b}}, the density estimator of the sampling distribution $\mu$,
    \[
        \hat{p}(x)=\frac{1}{nT}\sum_{t=1}^{T}\sum_{i=1}^{n_t}K_{\bar{h}}(x-x_{i,t}), \quad x \in [0, 1]^d,
    \]
satisfies,
\begin{align}
\label{Eq-5}
\vert\vert \hat{p}-\mathbb{E}(\hat{p})\vert\vert_{\infty}
\le 
C\sqrt{\frac{\log(nT)+\log(1/\bar{h})}{nT\bar{h}^d}}
\end{align}
with probability at least $1-\frac{1}{nT}$.
Moreover, under Assumption~\ref{Kernel-as}-{\bf{a}}, the bias term satisfies
\begin{align}
\label{Eq-4}
\vert\vert \mathbb{E}(\hat{p})-u\vert\vert_{\infty}\le C_2\bar{h}^r.
\end{align}
Therefore,
\begin{align}
\label{Eq-7}
\vert\vert \hat{p}-u\vert\vert_{\infty}=O\Big( \Big(\frac{\log(nT)}{nT}\Big)^{\frac{2r}{2r+d}}\Big)
\end{align}
with probability at least $1-\frac{1}{nT}$.
\end{lemma}
 The verification of these bounds can be found in many places in the literature. For equation \eqref{Eq-5} see for example  \cite{GINE2002907}, \cite{rinaldo2010generalized}, \cite{sriperumbudur2012consistency} and \cite{jiang2017uniform}. For equation \eqref{Eq-4}, \cite{Tsybakov:1315296} is a common reference.
\begin{proof} [Proof of \Cref{theorem:FSBS}] 
For any $(s,e]\subseteq (0,T] $, let 
	$$ \widetilde f^{  (s, e]}_{t }    (x ) =  \sqrt { \frac{e-t}{ (e -s  )(t-s )}} \sum_{ l =s +1}^{ t}  f_l  ^* (x )  -
   \sqrt { \frac{t-s}{ (e-s  )(e-t)}} \sum_{l =t+1}^{ e}  f_l  ^* (x ),\ x\in[0,1]^d.$$
   For any $\tilde{r}\in (\rho, T-\rho ] $ and $x\in [0,1]$, we consider
\begin{align*}
& 
\mathcal A_x (  (s,e], \rho , \lambda   )  =
\bigg\{ \max_{t=s+\rho+1}^{e-\rho } |\widetilde F_{t, h } ^{s,e } (x)   - \widetilde f_{t} ^{s,e }  (x)  | \le \lambda \bigg\} ;
\\
 &   
\mathcal B_x (  \tilde{r}  ,  \rho,  \lambda) = \bigg\{\max_{ N=\rho }^{T -\tilde{r}   } \bigg| \frac{1}{\sqrt  N} \sum_{t=\tilde{r}+1}^{\tilde{r}+ N}    F_{t,h   }(x)   - \frac{1}{\sqrt  N} \sum_{t=\tilde{r}  +1}^{\tilde{r}+N}   f_{t }  (x)  \bigg| \le \lambda \bigg\}  \bigcup 
\\
& 
\quad \quad  \quad \quad \quad \ \   \bigg\{\max_{ N=\rho }^{   \tilde{r} } \bigg| \frac{1}{\sqrt  N}  \sum_{t=\tilde{r} -N +1}^{\tilde{r}}  F_{t,h  }(x)   - \frac{1}{\sqrt  N}  \sum_{t=\tilde{r} -N +1}^{\tilde{r}}     f_{t}  (x)  \bigg| \le \lambda   \bigg\}.
\end{align*}
\\
\\
From \Cref{alg:FSBS}, we have that
$$ \rho =\frac{\log(T)}{nh^d}.$$
We observe that, 
$  \rho nh^d =  \log(T)  $ 
and for $T\ge 3,$ 
we have that
\begin{align*}  
\rho  ^{1/2 -1/q }  
\ge 
{(nh^d)^{ 1/2  - (q-1)/q }} .
\end{align*}
Therefore, \Cref{prop:deviation of means} and \Cref{cor:deviation cusum} imply that  with
\begin{align}\label{eq:size of lambda in functions} 
\lambda  
=
C_\lambda   \bigg(   \log ^{5/q   }  (T)        \sqrt {   \frac{1}{nh^d }  
   + 1  }  +    \sqrt { \frac{ \log(T) }{   nh^d   } }  +\sqrt{T}h^r+\sqrt{T}  \Big(\frac{\log(nT)}{nT}\Big)^{\frac{2r}{2r+d}} \bigg)   ,
\end{align} 
for some diverging sequence $C_\lambda $, it holds that 
$$P \bigg\{ \mathcal A _x^{c} ( (s,e], \rho,  \lambda )  \bigg\} \le 4C_1\frac{\log(T)}{(\log ^{5/q}  (T))^{q}}+\frac{2}{T^5}
+\frac{10}{Tn}$$
and
$$P \bigg\{ \mathcal B_x^{c} ( r  , \rho,   \lambda )  \bigg\}\le 2C_1\frac{\log(T)}{(\log ^{5/q}  (T))^{q}}+\frac{1}{T^5}
+\frac{5}{Tn}.$$
Then, using that $\log^4(T)=O(T),$ from above
$$P \bigg\{ \mathcal A _x^{c} ( (s,e], \rho,  \lambda )  \bigg\} =O(\log^{-4} (T)  ) \quad \text{and} \quad   P \bigg\{ \mathcal B_x^{c} ( r  , \rho,   \lambda )  \bigg\} =O(\log^{-4} (T) ) .
$$
Now, we notice that,
\begin{align*}
&\sum_{k=1}^ \mathcal K \tilde{n}_k 
=
\sum_{k=1}^ \mathcal K (2^k -1)
\le
\sum_{k=1}^ \mathcal K 2^k
\le
2(2^{\lc \log(2)C_{\mathcal{K}}(\log(\log(T)))/\log2 \rc}-1)
\\
&\le
4(2^{(\log(\log(T)))/\log2})^{\log(2)C_{\mathcal{K}}}
=O(\log^{\log(2)C_{\mathcal{K}}}((T)))
.
\end{align*}
In addition, there are $K =O(1) $ number of  change-points. In consequence, it follows that 
\begin{align}
& 
P \bigg\{ \mathcal A _u ( \mathcal I , \rho ,  \lambda )   \text{ for all } \mathcal I \in \mathcal J \text{ and all }u \in \{ u_m\}_{m=1}^{ \log(T)}   \bigg\} \ge1- \frac{1}{\log^2(T)} , \label{eq:event B0} 
\\
&
P \bigg\{ \mathcal B_u  (  s , \rho,  \lambda )  \cup \mathcal B_u (e, \rho  ,  \lambda )   \text{ for all } \mathcal (s,e]=\mathcal  I \in \mathcal J \text{ and all }u \in \{ u_m\}_{m=1}^{ \log(T)}  \bigg\} \ge 1- \frac{1}{\log(T)},  \label{eq:event B1} 
\\
&
P \bigg\{ \mathcal B_u  (  \eta_k  , \rho,  \lambda )   \text{ for all }  1\le k \le K  \text{ and all }u \in \{ u_m\}_{m=1}^{ \log(T)}  \bigg\} \ge1- \frac{1}{\log^3(T)}  \label{eq:event B2} .
\end{align} 
The rest of the argument is made by assuming the events in equations \eqref{eq:event B0}, \eqref{eq:event B1} and  \eqref{eq:event B2} hold. 
\\
\\
Denote 
$$\Upsilon_k  =C      \log^{ \max\{1, 10/q\} } (T)  \bigg(   1 +         T ^{ \frac{  d}{ 2r+d }  }  n ^{ \frac{-  2r }{ 2r+d }  }          \bigg) \kappa_k^{-2}      \quad \text{and} 
\quad \Upsilon_{\max  } = C      \log^{\max\{1, 10/q\}  } (T)  \bigg(   1 +         T ^{ \frac{  d}{ 2r+d }  }  n ^{ \frac{-  2r }{ 2r+d }  }          \bigg) \kappa ^{-2}    ,  $$
where $ \kappa = \min \{\kappa_1, \ldots, \kappa_K \} $.
Since $\Upsilon _k$ is  the desired   localisation rate, by induction, it suffices to consider any generic interval $(s, e]  \subseteq (0, T]$ that satisfies the following three conditions:
	\begin{align*}
		&
		\eta_{m-1} \le s\le \eta_m \le \ldots\le \eta_{m+q} \le e \le \eta_{m+q+1}, \quad q\ge -1; \\
		& 
		\text{ either }   \eta_m-s\le \Upsilon_m  \quad \text{or} \quad   s-\eta_{m-1} \le  \Upsilon_{m-1};
		\\
		& 
		\text{ either }  \eta_{m+q+1}-e \le \Upsilon_{m+q+1} \quad  \text{or}\quad   e-\eta_{m+q} \le \Upsilon _{m+q}  .
	\end{align*}
Here $q = -1$ indicates that there is no change-point contained in $(s, e]$.
\\
\\
Denote 
$$\Delta_k =  \eta_{k-1}-\eta_{k }  \text{ for }  k =1, \ldots, K+1  \quad \text{and} \quad \Delta= \min\{ \Delta_1,\ldots, \Delta_{K+1}\} .$$
Observe  that  since 
$\kappa_k >0$ 
for all $1\le k\le K$
and  that 
$\Delta_k =\Theta(T) $,   
it holds that $\Upsilon_{\max  }   =o(\Delta)  $.  
Therefore, it has to be the case that for any true  change-point $\eta_m\in (0, T ]$, 
either 
$|\eta_m  -s |\le \Upsilon_{m}  $ 
or 
$|\eta_m -s| \ge \Delta - \Upsilon_{\max  }     \geq  \Theta(T) $. This means that 
$ \min\{ |\eta_m-e|, |\eta_m-s| \}\le \Upsilon_{m}  $ 
indicates that 
$\eta_m$ 
is a detected change-point in the previous induction step, even if $\eta_m\in (s, e]$.  
We refer to 
$\eta_m\in (s,e]$ 
as an undetected change-point if 
$ \min\{ \eta_m -s, \eta_m-e\} =\Theta(T)  $.
To complete the induction step, it suffices to show that FSBS $( (s,e], h    ,  \tau   )$     
\\
{\bf  (i)} will not detect any new change
point in $(s,e ]$ if
all the change-points in that interval have been previously detected, and
\\
{\bf	(ii) }will find a point $D_{m*}^{\mathcal I^* }$ in $(s,e]$ such that $|\eta_m-D_{m*}^{\mathcal I^* } |\le \Upsilon_m$ if there exists at least one undetected change-point in $(s, e]$.
\\
\\
In order to accomplish this, we need the following series of steps.
\\
\\
 {\bf Step 1.}  We first observe that if $\eta_k \in \{ \eta_k\}_{k=1}^K$ is any change-point in the functional time series, by \Cref{lemma:properties of seeded},
 there exists  a seeded interval  $ \mathcal I_k =(s_k, e_k] $ containing  exactly one change-point $\eta_k $   such that   
\begin{align*}   
\min\{ \eta_k-s_k , e_k -\eta_k  \}\ge \frac{1}{16} \zeta_k ,  \quad \text{and}
 \quad \max\{ \eta_k-s_k , e_k -\eta_k  \}\le \zeta _k
 \end{align*}
where,
$$ \zeta _k = \frac{9}{10} \min \{ \eta_{k+1}-\eta_k, \eta_k -\eta_{k-1} \} . $$ 
Even more, we notice that if  
$ \eta_k\in (s,e]$ 
is any undetected change-point in 
$(s,e] $. 
Then it must hold that 
 $$ s-\eta_{k-1}  \le \Upsilon_{\max}  .$$
Since 
$\Upsilon_{\max} =  O( \log^{\max\{1, 10/q\}  } (T)           T ^{ \frac{  d}{ 2r+d }  } )$ and $O(\log^a(T))=o(T^b)$ for any positive numbers $a$ and $b$, we have that $\Upsilon_{\max}=o(T)$. Moreover,  
$\eta_k -s _k \le \zeta_k  \le \frac{9}{10}(\eta_k -\eta_{k-1})  $, so that it holds that 
$$s_k -\eta_{k-1} \ge \frac{1}{10} (\eta_{k} -\eta_{k-1} ) >  \Upsilon_{\max } \ge s- \eta_{k-1} $$
and in consequence 
$  s_k\ge s $. 
Similarly $  e_k\le e $.   
Therefore 
$$ \mathcal I_k = (s_k, e_k] \subseteq (s,e]. $$
\
\\
{\bf Step 2.} Consider the collection of intervals  $\{ \mathcal I _k =(s_k, e_k] \}_{k=1}^K $ in {\bf Step 1.}
 In this step, it is shown that   for each $k \in \{ 1,\ldots,K\}$, it holds that
	\begin{align} \label{eq:properties of vectors}
	 \max_{ t=s_k  +\rho}^{t=e _k  - \rho} 	\max_{m=1}^{m=\log(T) }| \widetilde F_{t, h }^{ ( s_k  , e_k  ]   } (u_m ) |  \ge   c_1 \sqrt {T}   \kappa_k   ,
	\end{align}
	for some sufficient small constant $c_1$. 
	 \\
	 \\
	 Let $k \in \{ 1,\ldots,K\}$.
	By  {\bf Step 1}, $ \mathcal I_k $
	contains exactly one change-point $\eta_k$.  Since  for every 
	$ u_m$, $ f^* _t (u_m)   $ is a    one dimensional population  time series and there is only  one change-point in $\mathcal I_k=(s_k, e_k]$, 
	it holds that 
$$f^*_{s_k+1}(u_m)=...=f^*_{\eta_k}(u_m)\neq f^*_{\eta_k+1}(u_m)=...=f^*_{e_k}(u_m)$$	
which implies, for $s_k<t< \eta_k$
\begin{align*}
\widetilde f^{  (s_k, e_k]}_{t }    (u_m ) =&  \sqrt { \frac{e_k-t}{ (e_k -s_k  )(t-s_k )}} \sum_{ l =s_k +1}^{ t}  f_{\eta_k}  ^* (u_m )  -\sqrt { \frac{t-s_k}{ (e_k-s_k  )(e_k-t)}} \sum_{l =t+1}^{\eta_k}  f_{\eta_k}  ^* (u_m )\\
-&
   \sqrt { \frac{t-s_k}{ (e_k-s_k  )(e_k-t)}} \sum_{l= \eta_k+1}^{ e_k}  f_{\eta_k+1}  ^* (u_m ) 
 \\
 =&
 (t-s_k)\sqrt { \frac{e_k-t}{ (e_k -s_k  )(t-s_k )}}f_{\eta_k}  ^* (u_m )
 -(\eta_k-t)\sqrt { \frac{t-s_k}{ (e_k-s_k  )(e_k-t)}}f_{\eta_k}  ^* (u_m )\\
-&
 (e_k-\eta_k)\sqrt { \frac{t-s_k}{ (e_k-s_k  )(e_k-t)}}f_{\eta_k+1}  ^* (u_m ) 
 \\
 =&
 \sqrt { \frac{(t-s_k)(e_k-t)}{ (e_k -s_k  )}}f_{\eta_k}  ^* (u_m )
 -(\eta_k-t)\sqrt { \frac{t-s_k}{ (e_k-s_k  )(e_k-t)}}f_{\eta_k}  ^* (u_m )\\
-&
 (e_k-\eta_k)\sqrt { \frac{t-s_k}{ (e_k-s_k  )(e_k-t)}}f_{\eta_k+1}  ^* (u_m ) 
 \\
 =&
 (e_k-t)\sqrt { \frac{t-s_k}{ (e_k-t)(e_k -s_k  )}}f_{\eta_k}  ^* (u_m )
 -(\eta_k-t)\sqrt { \frac{t-s_k}{ (e_k-s_k  )(e_k-t)}}f_{\eta_k}  ^* (u_m )\\
-&
 (e_k-\eta_k)\sqrt { \frac{t-s_k}{ (e_k-s_k  )(e_k-t)}}f_{\eta_k+1}  ^* (u_m ) 
 \\
 =&
 (e_k-\eta_k)\sqrt { \frac{t-s_k}{ (e_k-t)(e_k -s_k  )}}f_{\eta_k}  ^* (u_m )
-
 (e_k-\eta_k)\sqrt { \frac{t-s_k}{ (e_k-s_k  )(e_k-t)}}f_{\eta_k+1}  ^* (u_m ) 
 \\
 =&
  (e_k-\eta_k)\sqrt { \frac{t-s_k}{ (e_k-t)(e_k -s_k  )}}(f_{\eta_k}  ^* (u_m )-f_{\eta_k+1}  ^* (u_m ) ).
\end{align*}
Similarly, for $\eta_k\le t\le e_k$
\begin{align*}
 f^{  (s_k, e_k]}_{t }    (u_m)=\sqrt {\frac { e  _k-t}{(e  _k- s  _k)(t-s   _k)} }(\eta_k-s  _k)    (   f ^* _{ \eta_{k}  } (u_m )  -f^* _{ \eta_{k}+1    }(u_m) ).   
\end{align*}
Therefore,
\begin{align}\label{eq:one change-point 1d cusum}
\widetilde f^{  (s_k, e_k]}_{t }    (u_m)      
=
\begin{cases}
\sqrt {\frac {t-s  _k}{(e  _k-s  _k)(e  _k-t)} }( e  _k-\eta_k)    (   f ^* _{ \eta_{k}  } (u_m )  -f ^* _{ \eta_{k  }+1    }(u_m) ) ,  & s_k   <   t< \eta_k;  
\\
\sqrt {\frac { e  _k-t}{(e  _k- s  _k)(t-s   _k)} }(\eta_k-s  _k)    (   f ^* _{ \eta_{k} } (u_m )  -f^* _{ \eta_{k  }+1    }(u_m) )  , & \eta_k \le  t\le e_k .
\end{cases} 
\end{align}
\
\\
	By \Cref{lemma:discrete approximation of a function}, with probability at least $1-o(1)$,  there exists $ u_{\tilde{k}} \in \{u_m \}_{m=1}^{\log(T)}$ such that 
	$$  |   f ^* _{ \eta_{k}  } (u_{\tilde{k}} )  -f ^* _{ \eta_{k}+1    }(u_{\tilde{k}}) |\ge \frac{3}{4} \kappa_k. $$
	 Since $\Delta=\Theta(T)$, $\rho=O(\log(T)T^{\frac{d}{2r+d}})$ and $\log^a(T)=o(T^{b})$ for any positive numbers $a$ and $b,$ we have that
	 \begin{equation}
	 \label{eq-bound}
	     \min\{ \eta_k-s_k , e_k -\eta_k  \}\ge \frac{1}{16} \zeta_k  \ge c_2 T  >  \rho  ,
	 \end{equation}
   so that  $\eta_k \in [s_k+\rho, e_k-\rho]$.
   Then, from \eqref{eq:one change-point 1d cusum}, \eqref{eq-bound} and the fact that $\vert e_k-s_k\vert<T$ and $\vert \eta_k-s_k\vert<T$, 
   \begin{align}        \label{eq:functional population lower bound}  | \widetilde f^{  (s_k, e_k]}_{ \eta_k  }    (u_{\tilde{k}} )  | 
   =   \sqrt {\frac { e  _k-\eta_k}{(e  _k- s  _k)(\eta_k-s   _k)} }(\eta_k-s  _k)    \vert   f ^* _{ \eta_{k} } (u_{\tilde{k}} )  -f^* _{ \eta_{k  }  +1  }(u_{\tilde{k}}) \vert 
   \ge  c_ 2 \sqrt {T } \frac{3}{4}\kappa_k.
   \end{align} 
	\
	\\
	Therefore, it holds that 
		\begin{align*}
	\max_{ t=s_k  +\rho}^{t=e _k  - \rho} \max_{m=1}^{m=\log(T)} | \widetilde F^{  (s_k, e_k]}_{t,h  }    (u_{m} )  | 
	\ge & 	 | \widetilde F^{  (s_k, e_k]}_{\eta_k,h }    (u_{\tilde{k}} )  |
	\\
	\ge & | \widetilde f^{  (s_k, e_k]}_{ \eta_k  }    (u_{\tilde{k}})  | 
	 -\lambda 
	\\
	\ge &  c_ 2 \frac{3}{4} \sqrt {T} \kappa_k  -\lambda  ,
	   \end{align*} 
where the first inequality follows from the fact that $\eta_k \in [s_k+\rho, e_k-\rho]$, the second inequality follows from the good  event in \eqref{eq:event B0}, and the last inequality follows from \eqref{eq:functional population lower bound}. 
	\\
	Next, we observe that $\log^{\frac{5}{q}}(T)\sqrt{\frac{1}{nh^d}+1}=o(\sqrt{T^{\frac{2r+d}{d}}})O(\sqrt{T^{\frac{d}{2r+d}}})=o(\sqrt{T})$, $\rho<c_2T$, $h^r=o(1)$ and $\Big(\frac{\log nT}{nT}\Big)^{\frac{2r}{2r+d}}=o(1)$.
	In consequence, since $\kappa_k  $ is a positive constant, by the upper bound of $\lambda $ on  \Cref{eq:size of lambda in functions}, for sufficiently large $T$, it holds that 
	$$ \frac{c_2}{4}\sqrt{T}\kappa_k\ge \lambda.$$
	Therefore,
	$$\max_{ t=s_k  +\rho}^{t=e _k  - \rho}	\max_{m=1}^{m=\log(T) }| \widetilde F^{  (s_k, e_k]}_{t,h  }    (u_{m} )  |   \ge \frac{c_2}{2}\sqrt {T} \kappa_k. $$
	Therefore \Cref{eq:properties of vectors} holds with $ c_1 =\frac{c_2}{2}.$
\\	
 \\
 {\bf Step 3.} 
In this step, it is  shown that   FSBS$( (s, e] , h ,\tau)   $  can 
 consistently detect or reject the existence of undetected
change-points within $(s, e]$.
 \\
 \\
Suppose $\eta_k\in (s, e ]$ is any undetected change-point. Then by the second half of {\bf Step 1}, $\mathcal I_k \subseteq (s,e] $. Therefore 
  \begin{align*}  
		A^{\mathcal I^* }_{m^* } \ge \max_{ t=s_k  +\rho}^{t=e _k  - \rho}	\max_{m=1}^{m=\log(T) }| \widetilde F_{t, h }^{ ( s_k  , e_k  ]   } (u_m ) |  \ge   c_1 \sqrt {T}   \kappa_k   >  \tau,
	\end{align*} 
where the second inequality follows from    \Cref{eq:properties of vectors}, and the last inequality follows from the fact that, $\log^a(T)=o(T^b)$ for any positive numbers $a$ and $b$ implies $$\tau = C_\tau      \bigg(   \log ^{ \max\{1, 10/q\} }   (T)        \sqrt {   \frac{1}{nh^d }  
   + 1  }      \bigg)  =o(\sqrt{T})$$.
\\
\\
Suppose there does not exist any undetected change-point in $(s, e]$. Then for any $ \mathcal I =(\alpha ,\beta]  \subseteq (s, e]$, one of the following situations must hold,
	\begin{itemize}
	\item [(a)]	There is no change-point within $ (\alpha ,  \beta   ]$;
	\item [(b)] there exists only one change-point $\eta_{k} $ within $(\alpha , \beta ]$ and $\min \{ \eta_{k}- \alpha   ,  \beta-\eta_{k}\}\le \Upsilon_k $; 
	\item [(c)] there exist two change-points $\eta_{k} ,\eta_{k+1}$ within $(\alpha ,  \beta ]$ and 
	$$  \eta_{k}-  \alpha  \le \Upsilon_k  \quad \text{and} \quad    \beta   -\eta_{k+1}  \le \Upsilon_ {k+1}  .$$
	\end{itemize}
The calculations of (c) are  provided as the other two cases are similar and simpler.  Note that for any $ x\in [0,1]^d$,
it holds that 
$$ | f _{\eta_{k+1} }  ^*  (x  )   - f _{\eta_{k+1} +1 }^*    (x )  |   \le \| f _{\eta_{k+1} }  ^* -f _{\eta_{k+1} +1 }^*  \|_\infty =\kappa_{k+1}     $$
and similarly 
$$ | f _{\eta_{k } }  ^*  (x   )   - f _{\eta_{k } +1 }^*    (x  )  |   \le   \kappa_{k }. $$
By \Cref{lemma:cusum boundary bound} and the assumption that $(\alpha  ,\beta  ]$ contains only two change-points, it holds that
for all $ x\in [0,1]^d$,
\begin{align*}\nonumber 
			\max_{ t= \alpha}^{  \beta   }  |\widetilde{f}^{ (a , \beta ] }_t (x )  | \leq & \sqrt{ \beta    - \eta_{r+1}}  | f ^* _{\eta_{r+1} }  ( x  )   - f^*  _{\eta_{r+1} +1 } ( x  ) |  + \sqrt{\eta_r - \alpha }  |  f _{\eta_{r } }^*   (x  )   - f _{\eta_{r }^*  +1 }  (x   ) | 
 \\
	 \le   &  \sqrt  { \Upsilon_{k+1} }  \kappa_{k+1 } +  \sqrt { \Upsilon_{k }}  \kappa_{k  } 
	 \le  2\sqrt{C} \log^{ \max\{1/2, 5/q\} } (T)  \sqrt {    1 +         T ^{ \frac{  d}{ 2r+d }  }  n ^{ \frac{-  2r }{ 2r+d }  }            } .
	 \end{align*}
	 Thus 
	 \begin{align}\label{eq:population size with two change-points}
	 \max_{ t= \alpha}^{  \beta   }  \|\widetilde{f}^{ (a , \beta ] }_t   \|_\infty \le 2\sqrt{C} \log^{ \max\{1/2, 5/q\} }  (T)  \sqrt {    1 +         T ^{ \frac{  d}{ 2r+d }  }  n ^{ \frac{-  2r }{ 2r+d }  }}  . 
	 \end{align}
 Therefore in the  good  event  in \Cref{eq:event B0}, for any $ 1\le m \le \log(T)$ and any $\mathcal I  =(\alpha, \beta] \subseteq(s,e] $, it holds that 
\begin{align*} 
A_ m^{\mathcal I}  
=&  
\max_{  t=\alpha+ \rho }^{  \beta-\rho      }  | \widetilde F _{t,h} ^{( \alpha  ,  \beta]  } (u_  m )    | 
\\
\le &  
\max_{  t=\alpha+ \rho }^{  \beta-\rho      }    \| \widetilde f^{  (\alpha , \beta ]}_{t }   \|_\infty  + \lambda
\\
\le&  
2\sqrt{C} \log^{ \max\{1/2, 5/q\} }  (T)  \sqrt {    1 +         T ^{ \frac{  d}{ 2r+d }  }  n ^{ \frac{-  2r }{ 2r+d }  }            }     + \lambda ,
\end{align*}
where the first inequality follows from \Cref{eq:event B0}, and the last inequality follows from  \Cref{eq:population size with two change-points}.
Then,
\begin{align*}
&2\sqrt{C} \log^{ \max\{1/2, 5/q\} }  (T)  \sqrt {    1 +         T ^{ \frac{  d}{ 2r+d }  }  n ^{ \frac{-  2r }{ 2r+d }  }            }     + \lambda 
\\
=&2\sqrt{C}\log^{ \max\{1/2, 5/q\} }  (T) \sqrt{\frac{1}{nh^d}+1}\\
&+C_\lambda\log^{5/q}(T)\sqrt{\frac{1}{nh^d}+1}+C_\lambda\sqrt{\frac{\log(T)}{nh^d}}+C_\lambda\sqrt{T}h^r+C_\lambda\sqrt{T}\Big(\frac{\log nT}{nT}\Big)^{\frac{2r}{2r+d}}.
\end{align*}
We observe that $\sqrt{\frac{\log(T)}{nh^d}}=O(\log(T)^{1/2}\sqrt{\frac{1}{nh^d}+1})$. Moreover, 
$$\sqrt{T}h^r=\sqrt{T}\Big(\frac{1}{nT}\Big)^{\frac{r}{2r+d}}\le(T^{\frac{1}{2}-\frac{r}{2r+d}})\frac{1}{n^{\frac{r}{2r+d}}},$$
and given that, 
$$\frac{1}{2}-\frac{r}{2r+d}=\frac{d}{2(2r+d)},$$
we get,
$$\sqrt{T}h^r=o\Big(\log^{\max{1/2,5/q}}(T)\sqrt{\frac{1}{nh^d}+1}\Big).$$ Following the same line of arguments, we have that
$$\sqrt{T}\Big(\frac{\log nT}{nT}\Big)^{\frac{2r}{2r+d}}=T^{\frac{1}{2}-\frac{2r}{2r+d}}\log^{\frac{2r}{2r+d}}(T)=o\Big(\log T\sqrt{\frac{1}{nh^d}+1}\Big).$$
 Thus, by the choice of $\tau$, it holds that with sufficiently large constant $C_\tau $, 
\begin{align} 
\label{help}
A_ m^{\mathcal I} \le \tau \quad \text{for all } 1\le m \le \log(T) \quad  \text{and all} \quad    \mathcal I  \subseteq(s,e] .
\end{align}
As a result, FSBS $( (s,e], h    ,  \tau   )$     will correctly
reject  if $(s,e]$ contains no undetected change-points. 
\\
\
\\
 {\bf Step 4.} 
Assume that there exists an undetected change-point $\eta_{\tilde{k}}\in (s, e]$ such that 
$$ \min\{ \eta_{\tilde{k}} -s, \eta_{\tilde{k}}-e\} =\Theta(T).$$  
Let  $m^*$ and $\mathcal I^*$ be defined as in  FSBS $( (s,e], h   ,  \tau   )$ with 
$$\mathcal I^* =( \alpha ^*, \beta ^*].  $$
\\
To complete the  induction, it suffices to show that,   there exists a change-point $\eta_k\in  (s ,e ]$ such that 
$ \min\{ \eta_k -s, \eta_k-e\}= \Theta(T)$ and $|D_{m^*} ^{\mathcal I^*} -\eta_k|\le \Upsilon_k$.
\\
\\
Consider the uni-variate time series $$ F_{t,h}  (u_{m*} )  =  \frac{1}{n} \sum_{i=1}^{n} y_{t, i} K_h(u_{m*}-x_{t, i })    \quad \text{and} \quad  
f^* _t   (u_{m*})    \quad \text{ for all } 1\le t  \le T  .$$ 
\
\\
Since the collection of the change-points of the time series 
$\{f_t ^*  (u_{m*})\}_{ t \in \mathcal I^* }$ 
is a subset of that of 
$\{\eta_{k}\}_{k=0}^{K+1}\cap(s,e]$, 
we may apply \Cref{lemma:error bound 1d} to by setting 
$$\mu_t = F_{t,h}  (u_{m*} ) \quad \text{and}  \quad 
\omega_t = f^* _t   (u_{m*})  $$ 
on the interval 
$\mathcal I^*$. 
Therefore, it suffices to justify that all the assumptions of 
\Cref{lemma:error bound 1d} hold. 
\\
\\ 
In the following, $\lambda$ is used in \Cref{lemma:error bound 1d}.
Then  \Cref{eq:wbs noise 1} and \Cref{eq:wbs noise 2}  are directly consequence of 
 \Cref{eq:event B0}, \Cref{eq:event B1}, \Cref{eq:event B2}.      \\
We observe that, for any $\mathcal I =(\alpha, \beta ]\subseteq (s,e],$
\begin{align*}
 \max_{ t=\alpha^* +\rho}^{ \beta^* -\rho }  | \widetilde F_{t, h }^{ (\alpha ^*  , \beta^* ]   } (u_{m^*}  ) |   = A^{\mathcal I^*}_{m^*}  \ge A^{\mathcal I  }_m  = \max_{   t= \alpha  +\rho }^{  \beta   - \rho   }| \widetilde F_{t, h }^{ ( \alpha   , \beta    ]   } (u_m ) | 
\end{align*}
for all $m $.
By  {\bf Step 1} with $ \mathcal I_k =(s_k,e_k]$, it holds that $$\min\{ \eta_k-s_k , e_k -\eta_k  \}\ge \frac{1}{16} \zeta_k  \ge c_2 T  ,$$ 
Therefore  for all $ k\in \{ \tilde{k} : \min \{ \eta_{\tilde{k}}-s, e-\eta_{\tilde{k}}\} \ge c_2 T\} $,
$$ \max_{ t=\alpha^* +\rho}^{ \beta^* -\rho } | \widetilde F_{t, h }^{ (\alpha ^*  , \beta^* ]   } (u_{m^*}  ) |   \ge 	\max_{ t=s_k  +\rho,m=1}^{t=e _k  - \rho,m=\log(T) }| \widetilde F_{t, h }^{ ( s_k  , e_k  ]   } (u_m ) |  \ge   c_1 \sqrt {T}   \kappa_k, $$
where the last inequality follows from \Cref{eq:properties of vectors}.
Therefore
\Cref{eq:wbs size of sample} holds in \Cref{lemma:error bound 1d}. 
Finally,
\Cref{eq:wbs noise} is a direct consequence of  the choices that 
$$
h  =  C_h (Tn)^{\frac{-1}{2r+d }}    \quad  \text{and} \quad \rho =     \frac{ \log(T) }{nh^d}     .$$
\
\\
Thus, all the conditions in \Cref{lemma:error bound 1d} are met. So that, there exists a change-point $\eta_{k}$ of $\{f^* _t   (u_{m*}) \}_{ t \in \mathcal I^* } $, satisfying
\begin{equation}
\min \{ \beta ^* -\eta_k,\eta_k- \alpha^* \}    > cT  , \label{eq:coro wbsrp 1d re1}
\end{equation}
and
\begin{align*}
		| D ^{\mathcal I^* }_{m*}-\eta_{k}|\le  \max \{ C_3\lambda  ^2 \kappa_k ^{-2}
	    ,\rho  \} \le&    C _4      \log^{ \max\{ 10/q,1\} } (T)  \bigg(   1 +         \frac{1}{nh^d }      + T h^{2r } 
	    +T\bigg(\frac{\log(nT)}{nT}\bigg)^{\frac{4r}{2r+d}}\bigg) \kappa_k^{-2}    \\
	    \le&   C      \log^{ \max\{ 10/q,1\} } (T)  \bigg(   1 +         T ^{ \frac{  d}{ 2r+d }  }  n ^{ \frac{-  2r }{ 2r+d }  }          \bigg) \kappa_k^{-2}   
	   \end{align*} 
for sufficiently large constant $C $, where we have followed the same line of arguments than for the conclusion of \eqref{help}.
Observe that \\
{\bf i)} The change-points
of $\{f_t^*  ( u_{m^*})\}_{ t \in \mathcal I^*}  $ belong to $(s, e]\cap \{ \eta_k\}_{k=1}^K$; and
\\
{\bf ii)} \Cref{eq:coro wbsrp 1d re1}  and  $( \alpha^* , \beta^*  ]  \subseteq (s, e]$ imply that
	\[
	\min \{e-\eta_k,\eta_k-s\}  >  cT     \ge  \Upsilon_{\max }.
	\]
	As discussed in the argument before {\bf Step 1}, this implies that
	$\eta_k $ must be an undetected change-point of  $\{f_t^*  ( u_{m^*})\}_{ t \in \mathcal I^*}  $.
 \end{proof}
  
\section{Deviation bounds related to  kernels}
In this section, we deal with all the large probability events occurred in the proof of \Cref{theorem:FSBS}.
Recall that $ F_{t, h} (x) =  \frac{ \frac{1}{n} \sum_{i=1}^{n} y_{t, i} K_h(x-x_{t, i }) }{\hat{p}(x)}$,   and
 \begin{align*} 
 \widetilde F_{t, h}^{(s,e] } (x) =     \sqrt { \frac{e-t}{ (e-s )(t-s)}} \sum_{ l =s+1}^{ t} F_{ l , h} (x)  -
   \sqrt { \frac{t-s}{ (e-s )(e-t)}} \sum_{l =t+1}^{ e} F_{l , h} (x). 
 \end{align*} 
By assumption \ref{Kernel-as}, we have $\max_{l = 1}^{q}\|K^l\|_{\infty}=\max_{l = 1}^{q}\|K \|_{\infty}^l  < C_K$, where $C_K > 0$ is an absolute constant. Moreover, assumption \ref{assume: model assumption}b implies $\vert f^{*}_{t}(x)\vert<C_f$ for any $x\in{{[0,1]^d}},t\in{1,...,T.}$
 
\begin{proposition}  \label{prop:deviation of means}  Suppose that  \Cref{assume: model assumption} and \ref{Kernel-as}  hold,
 that $  \rho nh^d \ge  \log(T)  $  and that $T\ge 3.$
Then for any $x\in [0,1]^d  $
{\small{
\begin{align} 
\label{eq:deviation of means 1}
 &\mathbb{P}  \bigg (\max_{k=\rho}^ { T-\tilde{r}   } \bigg| \frac{1}{\sqrt  k} \sum_{t=\tilde{r}+1}^{\tilde{r}+ k}  \big( F_{t,h }  (x) -f_t^*(x)    \big)   \bigg| \ge \frac{2}{\tilde{c}}z  \sqrt {   \frac{1}{nh^d }  
   + 1  }  +  \frac{\tilde{C}_1}{\tilde{c}}\Big(\sqrt { \frac{ \log(T) }{   nh^d   } }\Big)    +\frac{\tilde{C}}{\tilde{c}}\sqrt{T}h^r+\frac{\bar{C}C_f}{\tilde{c}}\sqrt{T} \Big(\frac{\log(nT)}{nT}\Big)^{\frac{2r}{2r+d}}\bigg)\nonumber
   \\
   &\le 2C_1\frac{\log(T)}{z^q} +T^{-5}+\frac{5}{Tn};
   \\\label{eq:deviation of means 2}
 & \mathbb{P}   \bigg(  \max_{ k=\rho}^{ \tilde{r}    } \bigg|  \frac{1}{\sqrt  k} \sum_{t=\tilde{r}-k+1}^{\tilde{r}}  \big( F_{t,h } (x)-f_t^*(x)   \big)   \bigg|  \ge \frac{2}{\tilde{c}}z  \sqrt {   \frac{1}{nh^d }  
   + 1  }  + \frac{\tilde{C}_1}{\tilde{c}}\Big(\sqrt { \frac{ \log(T) }{   nh^d   } }\Big)    +\frac{\tilde{C}}{\tilde{c}}\sqrt{T}h^r+\frac{\bar{C}C_f}{\tilde{c}}\sqrt{T} \Big(\frac{\log(nT)}{nT}\Big)^{\frac{2r}{2r+d}}\bigg) \nonumber
\\   
&\le 2C_1\frac{\log(T)}{z^q} +T^{-5}+\frac{5}{Tn}.
\end{align} 
}}

\end{proposition}
\begin{proof}The proofs of \Cref{eq:deviation of means 1} and \Cref{eq:deviation of means 2} are the same. So only the proof of \Cref{eq:deviation of means 1} is presented.  
We define the events $E_1=\Big\{\vert\vert \hat{p}-u\vert\vert_\infty \le \bar{C}\Big(\Big(\frac{\log(Tn)}{Tn}\Big)^{\frac{2r}{2r+d}}\Big)\Big\}$ and $E_2=\Big\{\hat{p}\ge \bar{c},\ \bar{c}=\inf_{x}u(x)-\bar{C}\Big(\frac{\log(Tn)}{Tn}\Big)^{\frac{2r}{2r+d}}\Big\}$. 
Using {\Cref{lemma:KDE}}, especifically by equation \eqref{Eq-7}, we have that $P(E_1)\ge1-\frac{1}{nT}$. Then, we observe that in event $E_1$, for $x\in[0,1]^d$
$$\inf_{s}u(s)-\hat{p}(x)\le u(x)-\hat{p}(x)\le|u(x)-\hat{p}(x)|\le \bar{C}\Big(\frac{\log(Tn)}{Tn}\Big)^{\frac{2r}{2r+d}}$$
which implies $E1\subseteq{E_2}$. Therefore, $P(E_2^c)\le \frac{1}{nT}$.
\\
Now, for any $x$, observe that, by definition of $F_{t,h}$ and triangle inequality
\begin{align}
\label{FirstB}
 I=\max_{k=\rho}^{ T-\tilde{r}   }  & \frac{1}{\sqrt k } \bigg| \sum_{ t=\tilde{r}+1  } ^{\tilde{r}+k}    F_{t, h} (x)    - \sum_{ t=\tilde{r}+1  } ^{\tilde{r}+k}    f_t^* (x)   \bigg|    \nonumber
\\
\le  \max_{k=\rho}^{ T-\tilde{r}   }   &  \bigg|  \frac{1}{ \sqrt{k} }\sum_{ t=\tilde{r}+1  } ^{\tilde{r}+k}  \frac{1}{n}   \sum_{i=1 }^{n}   \bigg( \frac{ f^*  _t   (x_{t,i })  K_h(x-x_{t,i })}{\hat{p}(x)}   -       f_t^{*}(x)  \bigg)   \bigg| \nonumber
\\
 +\max_{k=\rho}^{ T-\tilde{r}   }  
&   \bigg|    \frac{1}{\sqrt k }\sum_{ t=\tilde{r}+1  } ^{\tilde{r}+k }  \frac{1}{n}   \sum_{i=1 }^{n}    \frac{ \xi_t (x_{t,i } )K_h(x-x_{t,i })}{\hat{p}(x)}  \bigg|  
\\
+ \max_{k=\rho}^{ T-\tilde{r}   }   &    \bigg|   \frac{1}{ \sqrt k  }\sum_{ t=\tilde{r}+1  } ^{\tilde{r}+k }  \frac{1}{n}   \sum_{i=1 }^{n} \frac{\delta_{t, i } K_h(x-x_{t,i })}{\hat{p}(x)}   \bigg| \nonumber
\\
&\hspace{-1.2cm}= I_1+I_2+I_3.   \nonumber
\end{align} 
In the following, we will show that $I_1\le I_{1,1}+I_{1,2}+I_{1,3}$, and that
\begin{enumerate}
  \item $\mathbb{P}  \Big(I_{1,1}\ge\frac{\tilde{C}_1}{\tilde{c}}\Big( \sqrt { \frac{ \log(T) }{   nh^d   } }\Big)\Big)\le \frac{1}{T^5}+\frac{1}{Tn}$,
  \item $\mathbb{P}  \Big(I_{1,2}\ge \frac{\tilde{C}}{\tilde{c}}\sqrt{T}h^r\Big)\le \frac{1}{Tn}$,
  \item $\mathbb{P}  \Big( I_{1,3}\ge\frac{\bar{C}C_f}{\tilde{c}}\sqrt{T} \Big(\frac{\log(nT)}{nT}\Big)^{\frac{2r}{2r+d}}\Big)\le \frac{1}{Tn}$,
  \item $\mathbb{P}  \Big(I_2\ge \frac{1}{\tilde{c}}z\sqrt{\frac{1}{nh^d}+1}\Big)\le \frac{C_1\log T}{z^q}+\frac{1}{Tn}$,
  \item $\mathbb{P}  \Big(I_3\ge \frac{1}{\tilde{c}}z\sqrt{\frac{1}{nh^d}+1}\Big)\le \frac{C_1\log T}{z^q}+\frac{1}{Tn}$,
\end{enumerate}
in order to conclude that,
\begin{align*}
    & \mathbb P  \bigg (I \ge 2z  \sqrt {   \frac{1}{nh^d }  
   + 1  }  +  \tilde{C}_1\Big(\sqrt { \frac{ \log(T) }{   nh^d   } }\Big)    +\frac{\tilde{C}}{\tilde{c}}\sqrt{T}h^r+\frac{\bar{C}C_f}{\tilde{c}}\sqrt{T} \Big(\frac{\log(nT)}{nT}\Big)^{\frac{2r}{2r+d}}\bigg)\\
   \le& \mathbb P\Big(I_{1,1}\ge\tilde{C}_1\Big( \sqrt { \frac{ \log(T) }{   nh^d   } }\Big)\Big)+
  \mathbb  P\Big(I_{1,2}\ge \frac{\tilde{C}}{\tilde{c}}\sqrt{T}h^r\Big)+
  \mathbb  P\Big( I_{1,3}\ge\frac{\bar{C}C_f}{\tilde{c}}\sqrt{T} \Big(\frac{\log(nT)}{nT}\Big)^{\frac{2r}{2r+d}}\Big)
   \\
   +&
   \mathbb P\Big(I_2\ge z\sqrt{\frac{1}{nh^d}+1}\Big)+
  \mathbb  P\Big(I_3\ge z\sqrt{\frac{1}{nh^d}+1}\Big)
   \\
   \le&
    2C_1\frac{\log(T)}{z^q} +T^{-5}+\frac{5}{Tn}.
\end{align*}
\
\\
{\bf Step 1.} The analysis for $I_1$ is done. We observe that, 
\begin{align*}
&\max_{k=\rho}^{ T-\tilde{r}   }  \bigg|  \frac{1}{ \sqrt{k} }\sum_{ t=\tilde{r}+1  } ^{\tilde{r}+k}  \frac{1}{n}   \sum_{i=1 }^{n}   \bigg( \frac{ f^*  _t   (x_{t,i })  K_h(x-x_{t,i })}{\hat{p}(x)}   -       f_t^{*}(x)  \bigg)   \bigg| 
\\
& \le \max_{k=\rho}^{ T-\tilde{r}   } \bigg|  \frac{1}{ \sqrt{k} }\sum_{ t=\tilde{r}+1  } ^{\tilde{r}+k}  \frac{1}{n}   \sum_{i=1 }^{n}   \bigg( \frac{ f^*  _t   (x_{t,i })  K_h(x-x_{t,i })}{\hat{p}(x)}   -       \frac{\int f_t^* (z)  K_h(x-z) d\mu(z)}{\hat{p}(x)}  \bigg)   \bigg| 
\\
&+ \max_{k=\rho}^{ T-\tilde{r}   } \bigg|  \frac{1}{ \sqrt{k} }\sum_{ t=\tilde{r}+1  } ^{\tilde{r}+k}  \frac{1}{n}   \sum_{i=1 }^{n}   \bigg(  \frac{\int f_t^* (z)  K_h(x-z) d\mu(z)}{\hat{p}(x)}  -    f_t^*(x)    \bigg)   \bigg| =I_{1,1}+\tilde{I}_1.
\end{align*}
{\bf Step 1.1} The analysis for $I_{1,1}$ is done. We note that the random variables 
$ \{    f^* _t   (x_{t,i })  K_h(x-x_{t,i })     \} _{1\le i \le n_t,1\le t \le N}$ are independent   distributed  with mean $\int f_t^* (z)  K_h( x -z ) d\mu(z)   $ and 
\begin{align*}
 \mathrm{Var} \big  (   f^* _t   (x_{t,i })  K_h(x-x_{t,i })    \big ) 
\le&  E \big\{  ( f^* _t ) ^2    (x_{t,i })  K_h^2 (x-x_{t,i })    \big\}  
\\
=&\int_{[0,1]^d} ( f^*_t)^2(z) \frac{1}{h^{2d} } K^2 \big(\frac{x-z}{h} \big)d \mu(z)
\\
\le & \frac{C_f^2}{h^d} \int_{[0,1]^d}  \frac{1}{h^d }K^2  \big(\frac{x-z}{h} \big)d\mu(z)
\\
= &  \frac{C_f^2}{h^d} \int_{[0,1]^d}   K^2  \big(u \big)d\mu(u)    
 < \frac{C_f^2 C_K^2 }{h ^d  } .
 \end{align*}
Since $ |f^*_t    (x_{t,i })  K_h(x -x_{t,i }  ) | \le C_f C_K h^{-d}$, by Bernstein inequality \cite{vershynin2018high}, we have that 
$$  \mathbb{P}   \bigg(    \bigg |  \frac{1}{   k  n } \sum_{ t=r+1  } ^{r+ k }     \sum_{i=1 }^n   f^*  _t   (x_{t,i })  K_h(x-x_{t,i })   - \int f_t^* (z)  K_h( x -z ) d\mu(z)     \bigg|  
\ge \tilde{C}_1 \bigg\{  \sqrt { \frac{ \log(T) }{ k  nh^d   } }  +  \frac{ \log(T )  }{    k nh^d   }    \bigg\}     
 \bigg)   \le   T    ^{-6} . $$
 Since $k nh^d  \ge  \log(T)  $ if $k \ge \rho$, with probability at most $ T^{-5}$, it holds that
 $$ \max_{k=\rho}^{ T-\tilde{r}   }  \bigg|  \frac{1}{ \sqrt{k }  n } \sum_{ t=r+1  } ^{r+k }    \sum_{i=1 }^n   \bigg(  f^*  _t   (x_{t,i })  K_h(x-x_{t,i })   -       \int f_t^* (z)  K_h(x-z) d\mu(z)  \bigg)   \bigg|  \ge \tilde{C}_1 \sqrt { \frac{ \log(T) }{   nh^d   } }.$$
 Therefore, using that $P(E^c_2)\le\frac{1}{Tn},$ we conclude 
 $$ \max_{k=\rho}^{ T-\tilde{r}   }  \bigg|  \frac{1}{ \sqrt{k }   } \sum_{ t=\tilde{r}+1  } ^{\tilde{r}+k } \frac{1}{n_t}   \sum_{i=1 }^{n_t}   \bigg(  \frac{f^*  _t   (x_{t,i })  K_h(x-x_{t,i }) }{\hat{p}(x)}  -       \frac{\int f_t^* (z)  K_h(x-z) d\mu(z)}{\hat{p}(x)}  \bigg)   \bigg|  \ge \frac{\tilde{C}_1}{\tilde{c}} \sqrt { \frac{ \log(T) }{   nh^d   } }$$
 with probability at most $T^{-5}+\frac{1}{nT}.$
\
\\
\\
{\bf Step 1.2} The analysis for $I_{1,2}$ and $I_{1,3}$ is done.
We observe that
\begin{align}
\tilde{I}_1=&\max_{k=\rho}^{ T-\tilde{r}   } \bigg|  \frac{1}{ \sqrt{k} }\sum_{ t=\tilde{r}+1  } ^{\tilde{r}+k}  \frac{1}{n}   \sum_{i=1 }^{n}   \bigg(  \frac{\int f_t^* (z)  K_h(x-z) d\mu(z)}{\hat{p}(x)}  -    f_t^*(x)    \bigg)   \bigg|  \nonumber
\\\label{Eq-1}
\le & \max_{k=\rho}^{ T-\tilde{r}   }\bigg|  \frac{1}{ \sqrt{k} }\sum_{ t=\tilde{r}+1  } ^{\tilde{r}+k}  \frac{1}{n}   \sum_{i=1 }^{n}   \bigg(  \frac{\int f_t^* (z)  K_h(x-z) d\mu(z)}{\hat{p}(x)}  -  \frac{f_t^*(x)u(x)}{\hat{p}(x)}  \bigg)   \bigg| 
\\\label{Eq-2}
+&\max_{k=\rho}^{ T-\tilde{r}   }\bigg|  \frac{1}{ \sqrt{k} }\sum_{ t=\tilde{r}+1  } ^{\tilde{r}+k}  \frac{1}{n}   \sum_{i=1 }^{n}   \bigg(\frac{f_t^*(x)u(x)}{\hat{p}(x)}-f_t^*(x)   \bigg)   \bigg|=I_{1,2}+I_{1,3}.
\end{align}
Then, we observe that 
\begin{align*}
I_{1,2}=&\bigg|  \frac{1}{ \sqrt{k} }\sum_{ t=\tilde{r}+1  } ^{\tilde{r}+k}  \frac{1}{n}   \sum_{i=1 }^{n}   \bigg(  \int f_t^* (z)  K_h(x-z) d\mu(z)  -  f_t^*(x)u(x)  \bigg)   \bigg| 
\\
&\le \frac{1}{ \sqrt{k} }\sum_{ t=\tilde{r}+1  } ^{\tilde{r}+k}  \frac{1}{n}   \sum_{i=1 }^{n}    \bigg| \int f_t^* (z)  K_h(x-z) d\mu(z)  -  f_t^*(x)u(x) \bigg|    
\\
&\le  \frac{1}{ \sqrt{k} }\sum_{ t=\tilde{r}+1  } ^{\tilde{r}+k}  \frac{1}{n}   \sum_{i=1 }^{n}  \tilde{C} h^r 
\\
&= \frac{1}{ \sqrt{k} }\sum_{ t=\tilde{r}+1  } ^{\tilde{r}+k}   \tilde{C} h^r 
\\
&=\sqrt{k}\tilde{C} h^r 
\end{align*}
where the second inequality follows from assumption \ref{Kernel-as}. Therefore, using event $E_2$, we can bound \eqref{Eq-1} by $\frac{\tilde{C}}{\tilde{c}}\sqrt{T}h^{r}$ with probability at least $1-\frac{1}{nT}.$ Meanwhile, for \eqref{Eq-2} we have that,

\begin{align}
I_{1,3}=&\max_{k=\rho}^{ T-\tilde{r}   }\bigg|  \frac{1}{ \sqrt{k} }\sum_{ t=\tilde{r}+1  } ^{\tilde{r}+k}  \frac{1}{n}   \sum_{i=1 }^{n}   \bigg(\frac{f_t^*(x)u(x)}{\hat{p}(x)}-f_t^*(x)   \bigg)   \bigg| \nonumber
\\\label{Eq-3}
\le&\max_{k=\rho}^{ T-\tilde{r}   }  \frac{1}{ \sqrt{k} }\sum_{ t=\tilde{r}+1  } ^{\tilde{r}+k}  \frac{1}{n}   \sum_{i=1 }^{n}   \vert f_t^*(x)\vert \bigg|\frac{u(x)-\hat{p}(x)}{\hat{p}(x)}\bigg|.  
\end{align}
Then, since in the event $E_1$, it is satisfies that
$$\vert\vert \hat{p}-u\vert\vert_\infty \le \bar{C}\Big(\Big(\frac{\log(Tn)}{Tn}\Big)^{\frac{2r}{2r+d}}\Big),\ \text{and} \ \hat{p}\ge \bar{c};$$
 we have that equation \eqref{Eq-3}, is bounded by 
$$ \max_{k=\rho}^{ T-\tilde{r}   }  \frac{1}{ \sqrt{k} }\sum_{ t=\tilde{r}+1  } ^{\tilde{r}+k}  \frac{1}{n}   \sum_{i=1 }^{n} \frac{\bar{C}C_f}{\tilde{c}}\Big(\frac{\log(nT)}{nT}\Big)^{\frac{2r}{2r+d}}\le\frac{\bar{C}C_f}{\tilde{c}}\sqrt{T}\Big(\frac{\log(nT)}{nT}\Big)^{\frac{2r}{2r+d}}$$
with probability at least $1-\frac{1}{nT}.$
\
\\
\\
{\bf Step 2.}  The analysis for $I_2$ and $I_3$ is done. For $1 \le t \le T $, let  $$ Z_t = \frac{1}{n} \sum_{i=1}^{n}  \xi_ t ( x_{t,i} ) K_h(x- x_{t, i})  \quad \text{and} \quad W_t = \frac{1}{n } \sum_{i=1}^{n}  \delta_ {t, i}   K_h(x- x_{t, i}) .
 $$
By \Cref{prop:functional kernel deviation} and event $E_2$, it holds that 
 \begin{align*}
 \mathbb{P}   \bigg\{ \max_{k=\rho}^{ T-\tilde{r}   } \bigg|   \frac{1}{\sqrt k  }   \sum_{ t=\tilde{r}+1  } ^{\tilde{r}+k } \frac{Z_t}{\hat{p}(x)}  \bigg |   \ge \frac{1}{\tilde{c}} z \sqrt {   \frac{1}{nh^d }  
   + 1  }         \bigg\}    \le  \frac{C_1\log(T) }{z^q   }+\frac{1}{nT}
 \end{align*} 
 and 
 \begin{align*}
 \mathbb{P}   \bigg \{  \max_{k=\rho}^{ T-\tilde{r}   }  \bigg| \frac{1}{\sqrt k  }  \sum_{ t=\tilde{r}+1  } ^{\tilde{r}+ k } \frac{W_t}{\hat{p}(x)}  \bigg | \ge\frac{1}{\tilde{c}} z   \sqrt {   \frac{1}{nh^d }  
   + 1  }       \bigg\} \le \frac{C_1  \log(T)}{z ^q}+\frac{1}{nT}.
 \end{align*} 
  The desired result follows from putting the previous steps together. 
\end{proof}

\begin{corollary} \label{cor:deviation cusum}Suppose that 
$  \rho nh^d \ge  \log(T) $  and that $T\ge 3$.  Then for $z>0$
{\small{
\begin{align*} 
 &\mathbb{P}   \bigg \{ \max_{t=s+\rho+1}^{e-\rho }   \bigg|  \widetilde F_{t, h }^{(s,e] } (x)   -   \widetilde f_{t}^{(s,e]} (x)          \bigg| \ge  \frac{4}{\tilde{c}}z  \sqrt {   \frac{1}{nh^d }  
   + 1  }  + \frac{2\tilde{C}_1}{\tilde{c}}\Big(\sqrt { \frac{ \log(T) }{   nh^d   } }\Big)    +\frac{2\tilde{C}}{\tilde{c}}\sqrt{T}h^r+\frac{2\bar{C}C_f}{\tilde{c}}\sqrt{T} \Big(\frac{\log(nT)}{nT}\Big)^{\frac{2r}{2r+d}}\bigg\} 
\\
&\le  2T ^{-5} + \frac{4C_1 \log(T)}{z ^q }  +10\frac{1}{Tn}. 
\end{align*} 
}}
\end{corollary} 
\begin{proof}
By definition of $\widetilde F_{t, h}^{(s,e]}$ and $\widetilde f_{t}^{(s,e]}$, we have that
\begin{align*}
    \bigg | \widetilde F_{t, h}^{(s,e] } (x) -\widetilde f_{t}^{(s,e]} (x)  \bigg|  
&\le
\bigg |\sqrt{\frac{e-t}{(e-s)(t-s)}}\sum_{l=s+1}^{t}(F_{l,h}(x)-f^*_l(x))\bigg |
\\
&+
\bigg |\sqrt{\frac{t-s}{(e-s)(e-t)}}\sum_{l=t+1}^{e}(F_{l,h}(x)-f^*_l(x))\bigg |.
\end{align*}
Then, we observe that, 
\begin{align*}
   \sqrt{\frac{e-t}{(e-s)(t-s)}}\le\sqrt{\frac{1}{t-s}} \ \text{if} \ s\le t, \ \text{and} \  \sqrt{\frac{t-s}{(e-s)(e-t)}}\le \sqrt{\frac{1}{e-t}} \ \text{if} \ t\le e.
\end{align*}
Therefore,
\begin{align*}
X=\max_{t=s+\rho+1}^{e-\rho }\bigg | \widetilde F_{t, h}^{(s,e] } (x) -\widetilde f_{t}^{(s,e]} (x)  \bigg|  
\le 
& \max_{t=s+\rho+1}^{e-\rho }  \bigg|   \sqrt { \frac{1}{  t-s }} \sum_{ l =s+1}^{ t}
\bigg (  F_{ l , h} (x)  -   \{ f_{ l }^{*} (x) \bigg )  \bigg| 
  \\
  +
  & 
 \max_{t=s+\rho+1}^{e-\rho } \bigg|   \sqrt { \frac{1}{  e- t }} \sum_{ l =t+1}^{ e }
\bigg (  F_{ l , h} (x)  -  f_{ l }^{*}(x) \bigg )  \bigg|
=X_1+X_2.
\end{align*} 
Finally, letting $\lambda=\frac{4}{\tilde{c}}z  \sqrt {   \frac{1}{nh^d }  
   + 1  }  + \frac{2\tilde{C}_1}{\tilde{c}}\Big(\sqrt { \frac{ \log(T) }{   nh^d   } }\Big)    +\frac{2\tilde{C}}{\tilde{c}}\sqrt{T}h^r+\frac{2\bar{C}C_f}{\tilde{c}}\sqrt{T} \Big(\frac{\log(nT)}{nT}\Big)^{\frac{2r}{2r+d}},$ we get that
\begin{align*}
\mathbb P(X\ge \lambda)\le& 
\mathbb P(X_1+X_2\ge \frac{\lambda}{2}+\frac{\lambda}{2})
\\
\le&
 \mathbb P(X_1\ge \frac{\lambda}{2})+\mathbb P(X_2\ge\frac{\lambda}{2})
\\
\le&   2T ^{-5} + \frac{4C_1 \log(T)}{z ^q }  +10\frac{1}{Tn},
\end{align*}   
where the last inequality follows from \Cref{prop:deviation of means}. 
\end{proof}

\newpage


\subsection{Additional Technical Results}\label{Heavy-Tail-cond}
The following lemmas provide lower bounds for 
  $$ Z_t = \frac{1}{n } \sum_{i=1}^{n}  \xi_ t ( x_{t,i} ) K_h(x- x_{t, i})  \quad \text{and} \quad W_t = \frac{1}{n} \sum_{i=1}^{n}  \delta_ {t, i}   K_h(x- x_{t, i}).
 $$ 
 They are a direct consequence of the temporal dependence and heavy-tailedness of the data considered in \Cref{assume: model assumption}.
 \begin{lemma}\label{prop:functional kernel deviation}
Let $\rho  \le T $ be such that  $\rho nh^d\ge\log(T)$ and $T\ge3.$
 Let $ N \in \mathbb Z^+$ be such that $ N \ge \rho $.
 \\
{\bf a.}  
 Suppose that for any $ q\ge  3 $ it holds that 
 \begin{align}\label{eq:delta conditon for function 2} \sum_{t=1} ^\infty  t^{ 1/2 - 1/q }  \mathbb{E}\big\{  \|   \xi_t -\xi_t^* \|_{\infty   } ^{q} \big\} ^{1/q }=\mathrm{O}(1).
\end{align}  
 Then 
 for any $z>0 $,
 \begin{align*}
 \mathbb{P}   \bigg\{ \max_{ k=\rho}^{ N   }  \bigg|  \bigg\{   \frac{1}{nh^d }  
   + 1 \bigg\} ^{ -1/2  }  \frac{1}{\sqrt k}   \sum_{t=1}^ k  Z_t  \bigg |   \ge z   \bigg\}    \le  \frac{C_1\log(T) }{z^q   }.
 \end{align*}
 {\bf b.}  
 Suppose that for some $ q\ge  3 $,
 \begin{align}\label{eq:delta conditon for vector 2} \sum_{t=1} ^\infty  t^{ 1/2 -1/q  }  \max_{i=1}^{  n}\big\{   \mathbb{E}  |    \delta_ {t,i}   - \delta_ {t,i} ^* |  ^{q} \big\} ^{1/q }< \mathrm{O}(1). 
\end{align}  Then 
 for any $w>0 $,
 \begin{align*}
 \mathbb{P}   \bigg\{ \max_{ k=\rho}^{  N }  \bigg|  \bigg\{   \frac{1}{nh^d }  
   + 1 \bigg\} ^{ -1/2  }  \frac{1}{\sqrt k }   \sum_{t=1}^ k  W _t  \bigg |   \ge w    \bigg\}   \le  \frac{C_1\log(T) }{w^q   }.
 \end{align*} 
  \end{lemma}
  
\begin{proof}   The proof of part {\bf b} is similar and simpler than that of part {\bf a}. For conciseness, only the proof of {\bf a} is presented. 
\\
\\
By \Cref{corollary:moment of sum of Z} and  \Cref{eq:delta conditon for function 2},  for all $J\in \mathbb Z^+$,  it holds that 
 \begin{align*}
  \mathbb{E}\bigg\{  \max_{k=1 }^{ J} | \sum_{t=1}^ k Z_t  |^q  \bigg\}  ^{1/q }   
 \le   J ^{1/2}C        \bigg\{  \bigg ( \frac{1}{nh^d } \bigg)^{1/2 }     + 1 \bigg\}      
 +  J  ^{1/q } C''       \bigg\{  \bigg ( \frac{1}{nh^d } \bigg)  
 ^{ (q-1)/q   }+ 1 \bigg\} .
 \end{align*}
 \
 \\
 As a result there exists a constant $C_1$ such that 
\begin{align*}
   \mathbb{E}\bigg\{  \max_{k=1 }^{ J} | \sum_{t=1}^ k Z_t  |^q  \bigg\}   
  \le &
  C_1  J^{q /2}   \bigg\{  \bigg ( \frac{1}{nh^d } \bigg)^{1/2 }     + 1 \bigg\}      ^q + C_1  J     \bigg\{  \bigg ( \frac{1}{nh^d } \bigg)  
 ^{ (q-1)/q   }+ 1 \bigg\} ^q.
  \end{align*}
 We observe that 
 \begin{align}
 J^{q/2}=&\frac{q}{2}\int_{0}^{J}x^{q/2-1}dx
 \\
 =&\frac{q}{2}\Big(\int_{0}^{1}x^{q/2-1}dx+\int_{1}^{J}x^{q/2-1}dx\Big)
 \\
 \le&\frac{q}{2}\Big(1+\int_{1}^{J}x^{q/2-1}dx\Big)
 \\
 =&\frac{q}{2}\Big(1+\int_{1}^{2}x^{q/2-1}dx+...+\int_{J-1}^{J}x^{q/2-1}dx\Big)
 \\
 \le&\frac{q}{2}\Big(1+\int_{1}^{2}2^{q/2-1}dx+...+\int_{J-1}^{J}J^{q/2-1}dx\Big)
 \\
 =&\frac{q}{2}\sum_{k=1}^J k^{q/2-1},
 \end{align}
 which implies, there is a constant $C_2$ such that
 \begin{align*}
 & C_1  J^{q /2}   \bigg\{  \bigg ( \frac{1}{nh^d } \bigg)^{1/2 }     + 1 \bigg\}      ^q + C_1  J     \bigg\{  \bigg ( \frac{1}{nh^d } \bigg)  
 ^{ (q-1)/q   }+ 1 \bigg\} ^q
 \le C_2 \sum_{k=1}^J \alpha_k,
  \end{align*}
 where 
 $$\alpha_k = k^{q/2- 1} \bigg\{  \bigg ( \frac{1}{nh^d } \bigg)^{1/2 }     + 1 \bigg\}      ^q  +   \bigg\{  \bigg ( \frac{1}{nh^d } \bigg)  
 ^{ (q-1)/q   }+ 1 \bigg\} ^q  . $$
By theorem B.2 of \citeauthor{kirch2006resampling} (\citeyear{kirch2006resampling}),
\begin{align*}
 \mathbb{E} \bigg\{  \max_{k=1}^{ N }    \bigg|      \frac{1}{\sqrt  k }  \sum_{t=1}^ k   Z_t  \bigg | \bigg\} ^{ q } 
\le &   
4C_2 \sum_{ l=1}^N  l^{-q/2 } \alpha_l 
\\
= & 4C_2  \sum_{ l=1}^N \bigg( 
l^{ - 1} \bigg\{  \bigg ( \frac{1}{nh^d } \bigg)^{1/2 }     + 1 \bigg\}      ^q  +   l^{-q/2}  \bigg\{  \bigg ( \frac{1}{nh^d } \bigg)  
 ^{ (q-1)/q   }+ 1 \bigg\} ^q  \bigg) 
 \\
 \le 
 & C_3 \log(N) \bigg\{  \bigg ( \frac{1}{nh^d } \bigg)^{1/2 }     + 1 \bigg\} ^q   + C_3 N^{-q/2 + 1 }\bigg\{  \bigg ( \frac{1}{nh^d } \bigg)  
 ^{ (q-1)/q   }+ 1 \bigg\} ^q   
\end{align*} 
where the last inequality follows from the fact that $\int_1^N \frac{1}{x}=log(N)$ and that $\int_1^N x^\frac{-q}{2}=O(N^{-q/2+1}).$
Since 
$$N ^{1/2 -1/q }  \ge  \rho ^{1/2 -1/q }  \ge {(nh^d)^{ 1/2  - (q-1)/q }} , $$
it holds that, 
$
    \frac{1}{nh^d}\le N.
$
Moreover,
\begin{align*}
    N^{-q/2+1}\bigg\{  \bigg ( \frac{1}{nh^d } \bigg)  
 ^{ (q-1)/q   }+ 1 \bigg\}^q
 =&
 N^{-q/2+1}\bigg\{  \bigg ( \frac{1}{nh^d } \bigg)  
 ^{ (q-1)/q  +1/2-1/2 }+ 1 \bigg\}^q
 \\
 =&
 N^{-q/2+1}\bigg\{  \bigg ( \frac{1}{nh^d } \bigg)  
 ^{ (q-1)/q  -1/2 }\bigg ( \frac{1}{nh^d } \bigg)  
 ^{1/2 }+ 1 \bigg\}^q
 \\
 \le& N^{-q/2+1}\bigg\{ \bigg ( \frac{1}{nh^d } \bigg)  
 ^{1/2 }+ 1 \bigg\}^q
  \bigg\{\bigg ( \frac{1}{nh^d } \bigg)  
 ^{ (q-1)/q  -1/2 }+1 \bigg\}^q
 \\
 =&
  N^{-q/2+1}\bigg\{ \bigg ( \frac{1}{nh^d } \bigg)  
 ^{1/2 }+ 1 \bigg\}^q
  \bigg\{\bigg ( \frac{1}{nh^d } \bigg)  
 ^{1/2-1/q }+1 \bigg\}^q
 \\
 =&
  N^{-q/2+1}\bigg\{ \bigg ( \frac{1}{nh^d } \bigg)  
 ^{1/2 }+ 1 \bigg\}^q
  \bigg\{\bigg ( \frac{1}{nh^d } \bigg)  
 ^{ (q-2)/(2q)}+1 \bigg\}^q
 \\
 \le& C^{'}_3N^{-q/2+1}\bigg\{ \bigg ( \frac{1}{nh^d } \bigg)  
 ^{1/2 }+ 1 \bigg\}^q
  \bigg\{\bigg ( \frac{1}{nh^d } \bigg)  
 ^{ (q-2)/(2q)}\bigg\}^q
 \\
 =& C^{'}_3N^{-q/2+1}\bigg\{ \bigg ( \frac{1}{nh^d } \bigg)  
 ^{1/2 }+ 1 \bigg\}^q
  \bigg\{\bigg ( \frac{1}{nh^d } \bigg)  
 ^{ (q-2)/(2)}\bigg\}
 \\
 \le& C^{'}_3N^{-q/2+1}\bigg\{ \bigg ( \frac{1}{nh^d } \bigg)  
 ^{1/2 }+ 1 \bigg\}^q
  \bigg\{\bigg ( \frac{1}{nh^d } \bigg)  
 ^{ q/2-1}\bigg\}
 \\
 \le& C^{'}_3N^{-q/2+1}\bigg\{ \bigg ( \frac{1}{nh^d } \bigg)  
 ^{1/2 }+ 1 \bigg\}^q
  N  
 ^{ q/2-1}
 \\
 =&C^{'}_3\bigg\{ \bigg ( \frac{1}{nh^d } \bigg)
 ^{1/2 }+ 1 \bigg\}^q.
\end{align*}
It follows that,
$$  \mathbb{E} \bigg\{  \max_{k=1}^{ N }     \bigg|      \frac{1}{\sqrt  k }  \sum_{t=1}^ k   Z_t  \bigg | \bigg\} ^{ q } 
\le     C_4  \log(N) \bigg\{  \bigg ( \frac{1}{nh^d } \bigg)^{1/2 }     + 1 \bigg\}  ^q.    $$
   By Markov's inequality, for any $z>0$ and the assumption that $ T\ge N,$
   $$\mathbb{P}   \bigg\{ \max_{k=1}^{ N }   \bigg\{   \frac{1}{nh^d }  
   + 1 \bigg\} ^{ -1/2  }  \bigg|    \frac{1}{\sqrt k }   \sum_{t= 1  }^ k Z_t  \bigg |   \ge z   \bigg\}  \le  \frac{C_1\log(T) }{z^q   } .$$
   Since $N \ge \rho $, this directly implies that 
    \begin{align*}
 \mathbb{P}   \bigg\{ \max_{ k=\rho}^{ N   }  \bigg\{   \frac{1}{nh^d }  
   + 1 \bigg\} ^{ -1/2  }   \bigg|  \frac{1}{\sqrt k}   \sum_{t=1}^ k  Z_t  \bigg |   \ge z   \bigg\}    \le  \frac{C_1\log(T) }{z^q   }.
 \end{align*}
\end{proof}

 \begin{lemma}  \label{lemma:momnents of kernel and xi}
 Suppose \Cref{assume: model assumption} {\bf c} holds and $q\ge 2$. Then there exists   absolute constants $C>0$ so that
 \begin{align} 
 \label{eq:moments of Z 1}
 & \mathbb{E}  | Z_t  -  Z _t ^*  | ^q    \le     C    \mathbb{E}   \big\{   \|    \xi _ {t}    - \xi _ {t}  ^*\|_{ \infty  } ^{q}        \big\}  \bigg\{   \bigg( \frac{1}{nh^d }  \bigg) ^{q-1 }    + 1\bigg\}.
 \end{align}
 If in addition $\mathbb{E}\big\{\|\xi_ t\|_{\infty}^{q}\big\}=O(1)$,
 then there exists   absolute constants $C'$ such that
 \begin{align}
 \label{eq:moments of Z 2}
  \mathbb{E} |Z_t|^q        \le   C'   \bigg\{   \bigg ( \frac{1}{nh^d  } \bigg)  ^{q-1 }+ 1  \bigg\}.
 \end{align}
 \end{lemma}
 \begin{proof}  
 The proof of the \Cref{eq:moments of Z 2} is simpler  and simpler than \Cref{eq:moments of Z 1}. So only the proof of \Cref{eq:moments of Z 1} is presented. Note that since
 $ \{x_t \}_{t=1}^T$ and $\{ \xi_t\}_{t=1}^T $ are independent, and that $\{ x_t\}_{t=1}^T $ are independent identically distributed, 
 $$Z_t^* =  \frac{1}{n } \sum_{i=1}^{n}  \xi_ t ^*  ( x_{t,i} ) K_h(x- x_{t, i}). $$
{\bf Step 1.} Note that, by the Newton’s binomial    
\begin{align*} 
\mathbb{E}  | Z_t   -Z_t^*   | ^q   
=& \mathbb{E} \bigg \{  \bigg|  \frac{1}{n } \sum_{i=1}^{n}    \{ \xi _t^* -  \xi _t \}   ( x_{t,i} )  K_h(x- x_{t, i})   \bigg| ^q  \bigg\}    
\\  
\le&\frac{1}{{n}^q  }   \mathbb{E} \bigg \{   \sum_{ \substack{ \beta _1 + \beta _2+ \ldots+\beta _{n}  =q \\\beta _1\ge 0, \ldots, \beta _{n}\ge 0 }  }   {q  \choose  \beta _1, \beta _2, \ldots, \beta_ {n}}\prod_{ j=1}^{n} \big |      \{ \xi _t^* -  \xi _t \}    ( x_{t , {i  } } )  K_h(x- x_{t, { i } })   \big |   ^{\beta _j } \bigg  \}   
\\
=&\frac{1}{n^q  }   \mathbb{E} \bigg \{   \sum_{k=1}^q  \sum_{ \substack{ \beta _1 + \beta _2+ \ldots+ \beta _{n}  =q 
\\  
\\
\beta= (\beta_1,\ldots, \beta_{n}), \|\beta \|_0 = k , \beta\ge 0 } }   {q  \choose   \beta_1, \beta _2, \ldots,\beta _{n} }\prod_{ j=1}^{n} \big |      \{ \xi _t^* -  \xi _t \}    ( x_{t , {i  } } )  K_h(x- x_{t, { i } })   \big |   ^{\beta _j }  \bigg  \} .
\end{align*}
\
\\
{\bf Step 2.}  
For a fixed $\beta =(\beta_1,\ldots, \beta_{n})$ such that $\beta_1+\ldots+\beta_{n}=q  $ and that
$\| \beta\|_0=k $, consider
$$  \mathbb{E} \bigg \{ \prod_{ j=1}^{n} \big |      \{ \xi _t^* -  \xi _t \}    ( x_{t , {i  } } )  K_h(x- x_{t, { i } })   \big |   ^{\beta _j }  \bigg  \}.$$
Without loss of generality, assume that $ \beta_1,\ldots, \beta_k $ are non-zero. Then it holds that 
\begin{align*}
 &   \mathbb{E} \bigg \{     \big  |     (\xi _t^* -  \xi _t ) ( x_{t , { 1 } } )  \big| ^{\beta_1 }        \big  |K_h  (x- x_{t, {   1} })   \big |^{\beta_1 }  \cdots 
     \big  |    (\xi _t^* -  \xi _t) ( x_{t , { k } } )  \big| ^{\beta_k }     \big  |K_h  (x- x_{t, {  k} })   \big |^{\beta_k  }   \bigg  \}  
    \\
     = &  \mathbb{E} _\xi  \bigg\{  \int  \big  |   (\xi _t^* -  \xi _t)( r )   \big| ^{\beta_1 }        \big  |K_h  (x-  r )   \big |^{\beta_1 } d \mu(r) \cdots 
     \int  \big  |    ( \xi _t^* -  \xi _t)( r  )   \big| ^{\beta_k }         \big  |K_h  (x-  r )   \big |^{\beta_k  } d \mu(r)    \bigg\} 
     \\
     = &  \mathbb{E} _\xi  \bigg\{  \int  \big  |    (\xi _t^* -  \xi _t)( x- sh )  \big| ^{\beta_1 }     \frac{ \big  |K   (  s )   \big |^{\beta_1 }   }{   h^{ d(\beta_1-1) } }   d \mu(s) \cdots 
    \int  \big  |     ( \xi _t^* -  \xi _t)( x- sh )  \big| ^{\beta_k }    \frac{ \big  |K   (  s )   \big |^{\beta_k }   }{   h^{ d(\beta_k-1) } }   d \mu(s)    \bigg\}  
    \\
   \le  &  h^{-d \sum_{j=1}^k (\beta_j - 1 )}  \mathbb{E} _\xi  \bigg\{   \|     \xi _t^* -  \xi _t   \|_{  \infty  } ^{\beta_1 }  C_K^{\beta_1 } \cdots 
  \|    \xi _t^* -  \xi _t  \|_{ \infty  } ^{\beta_k }  C_K^{\beta_k }      \bigg\}   
  \\
  \le &  h^{-d  (q-k) } C_K^q   \mathbb{E} _\xi  \bigg\{   \|     \xi _t^* -  \xi _t   \|_{\infty  } ^{ \sum_{j=1}^k \beta_k}        \bigg\}   
  \\
  \le &h^{-d  (q-k) } C_K^q   \mathbb{E} _\xi  \big\{   \|     \xi _t^* -  \xi _t   \|_{\infty  } ^{q}        \big\}   
  \end{align*}
where the third equality follows by using the change of variable $s=\frac{x-r}{h},$ the first inequality by assumption~\ref{Kernel-as}.
\
\\
\\
{\bf Step 3.} 
Let $k\in \{ 1,\ldots, q\}$ be fixed.    Note that $ {q  \choose   \beta_1, \beta _2, \ldots,\beta _{n} }  \le q! $.  Consider  set 
$$ \mathcal B_k = \bigg\{ \beta \in\mathbb  N ^{n} : \beta\ge 0 ,  \beta_1+\ldots+\beta_{n} =q , \vert \beta\vert_0= k    \bigg \}.$$
To bound the cardinality of  the set $\mathcal B_k $,  first note that  since $\vert \beta\vert_0=k $, there are   
$ { n \choose  k}  $ number of ways to choose the index of non-zero entries of $\beta$. 
\\
Suppose $\{ i_1, \ldots i_k\} $ are the chosen index such that $\beta_{i_1}\not = 0, \ldots, \beta_{i_k} \not = 0. $
Then the constrains $\beta_{i_1}> 0, \ldots, \beta_{i_1}> 0 $ and  $ \beta_{i_1}+\ldots+\beta_{i_k} =q   $ are equivalent to that of  diving $q$ balls into $k$ groups (without distinguishing each ball).
As a result  there are 
${ q-1 \choose  k-1} $ number of ways to choose the $\{ \beta_{i_1}, \ldots, \beta_{i_k}\}$ once the index  $\{ i_1, \ldots i_k\} $    are chosen. 
\\
\\
{ \bf Step 4.} 
Combining the previous three steps, it follows that for some constants  $C_q, C_1>0$ only depending on $q$, 
\begin{align*} 
\mathbb{E}  | Z_t   -Z_t^*   | ^q   
\le&
\frac{1}{n^q  }   \mathbb{E} \bigg \{   \sum_{k=1}^q  \sum_{ \substack{ \beta _1 + \beta _2+ \ldots+ \beta _{n}  =q \\  
\\
\beta= (\beta_1,\ldots, \beta_{n}), \vert\beta \vert_0 = k , \beta\ge 0 } }{q  \choose   \beta_1, \beta _2, \ldots,\beta _{n} }\prod_{ j=1}^{n} \big |      ( \xi _t^* -  \xi _t )    ( x_{t , {i  } } )  K_h(x- x_{t, { i } })   \big |   ^{\beta _j }  \bigg  \} 
\\
\le& 
\frac{1}{n^{q}}       \sum_{k=1}^q   { n \choose k   } {q- 1 \choose k-1  } q!  h^{-d(q-k) } C_K^q   \mathbb{E} _\xi  \big\{   \| \xi _t^* -  \xi _t   \|_{\infty  } ^{q}        \big\}   
\\
\le& 
\frac{1}{n^{q}}       \sum_{k=1}^q  n^{k}  C_q C_K^q h^{-d(q-k) }    \mathbb{E} _\xi  \big\{   \|     \xi _t^* -  \xi _t   \|_{\infty  } ^{q}        \big\}   
\\
\le& 
C_ 1  \mathbb{E} _\xi  \big\{   \|      \xi _t^* -  \xi _t   \|_{\infty  } ^{q}        \big\} \bigg\{   \bigg( \frac{1}{nh^d }  \bigg) ^{q-1 } + \bigg( \frac{1}{nh^d }  \bigg) ^{q-2} + \ldots + \bigg( \frac{1}{nh^d }  \bigg)   + 1\bigg\}  
\\
\le&  
C_1    \mathbb{E} _\xi  \big\{   \|   \xi _t^* -  \xi _t   \|_{ \infty  } ^{q}        \big\}  q  \bigg\{   \bigg( \frac{1}{nh^d }  \bigg) ^{q-1 }    + 1\bigg\}  ,   
\end{align*}
where the second inequality is satisfied by step 3 and that ${q  \choose   \beta_1, \beta _2, \ldots,\beta _{n} }\le q!$, while the third inequality is achieved by using that ${ n \choose k   } {q- 1 \choose k-1  }q!\le { n \choose k   } C_q\le n^kC_q.$ Moreover, given that $\frac{1}{n^q}n^kh^{-d(q-k)}=\Big(\frac{1}{nh^d}\Big)^{q-k}$ the fourth inequality is obtained. The last inequality holds because if $\frac{1}{nh^d } \le 1  $, then 
$\bigg\{   \bigg( \frac{1}{nh^d }  \bigg) ^{q-1 } + \ldots + \bigg( \frac{1}{nh^d }  \bigg)   + 1\bigg\}  \le  q ,$ 
and if 
$\frac{1}{nh^d } \ge 1  $, 
then 
$ \bigg\{   \bigg( \frac{1}{nh^d }  \bigg) ^{q-1 } + \ldots + \bigg( \frac{1}{nh^d }  \bigg)   + 1\bigg\}  \le  q\bigg( \frac{1}{nh^d }  \bigg) ^{q-1 }$.
\end{proof} 
\begin{lemma} 
\label{corollary:moment of sum of Z}    
Suppose  \Cref{assume: model assumption} {\bf c} holds.  Let $\rho  \le T $ be such that  $\rho nh^d\ge\log(T)$ and $T\ge3.$
 Let $ N \in \mathbb Z^+$ be such that $ N \ge \rho $. Then, it holds that 
$$ \bigg\{  \mathbb{E}  \max_{k=1}^{N } | \sum_{t=1}^ k  Z_t  |^q  \bigg\} ^{1/q}    \le N ^{1/2}C        \bigg\{  \bigg ( \frac{1}{nh^d  } \bigg) ^{1/2 }    + 1 \bigg\}      
 +  N ^{1/q } C'       \bigg\{  \bigg ( \frac{1}{nh^d  } \bigg)  ^{ (q-1) /q  }+ 1 \bigg\}.$$
\end{lemma}  
\begin{proof}
We have that $q>2$ and $ \mathbb E|Z_1|<\infty$ by the use of \Cref{lemma:momnents of kernel and xi}. Then, making use of Theorem 1 of \citeauthor{liu2013probability} (\citeyear{liu2013probability}), we obtain that
\begin{align*}
\bigg\{ \mathbb{E}   \max_{k=1}^{N } | \sum_{t=1}^ k Z_t  |^q   \bigg\}  ^{1/q} 
\le&  
N ^{1/2}C_1 \bigg\{ \sum_{ j =1}^ N  \Theta_{j,2 }   +  \sum_{j  =N +1}^\infty \Theta_{j,q }     +  \{  \mathbb{E} |Z_1| ^2\}^{1/2 } \bigg\}
\\
+& 
N^{1/q } C_2 \bigg\{ \sum_{j  =1}^N  j^{1/2 -1/q  } \Theta_{j,q }    + \{  \mathbb{E} |Z_1|^q\}^{1/q } \bigg\},
\end{align*}
where $ \Theta_{j, q} =  \{  \mathbb{E}( |Z_j^* -Z_j  |^q) \} ^{1/q }$. Moreover, we observe that since 
$ \Theta_{j,2 }\le \Theta_{j,q}$ for any $ q\ge 2$,
it follows
\begin{align*}
\bigg\{ \mathbb{E}   \max_{k=1}^{N } | \sum_{t=1}^ k Z_t  |^q   \bigg\}  ^{1/q} 
\le&  
N ^{1/2}C_1 \bigg\{ \sum_{ j =1}^\infty \Theta_{j,q }    +  \{  \mathbb{E} |Z_1| ^2\}^{1/2 } \bigg\}
\\
+& 
N^{1/q } C_2 \bigg\{ \sum_{j  =1}^\infty  j^{1/2 -1/q  } \Theta_{j,q }    + \{  \mathbb{E} |Z_1|^q\}^{1/q } \bigg\},
\end{align*}
Next, by the first part of \Cref{lemma:momnents of kernel and xi},
$$ \Theta_{j, q} ^q  \le   C    \mathbb{E}    \big\{   \|    \xi _ {j}    - \xi _ {j }  ^*\|_{ \infty  } ^{q}        \big\}  \bigg\{   \bigg( \frac{1}{nh^d }  \bigg) ^{q-1 }    + 1\bigg\}.$$
even more, we have that $N\ge \frac{1}{nh^d}$, implies  that 
{\small{
\begin{align*}
\bigg\{  \mathbb{E}  \max_{k=1}^{N } | \sum_{t=1}^ k  Z_t  |^q  \bigg\}  ^{1/q} 
\le&
N ^{1/2}C^{'}_1 \bigg\{ \sum_{ j =1}^\infty C    \mathbb{E}    \big\{   \|    \xi _ {j}    - \xi _ {j }  ^*\|_{ \infty  } ^{q}        \big\}  \bigg\{   \bigg( \frac{1}{nh^d }  \bigg) ^{q-1 } \bigg\}^{1/q}    +  \{  \mathbb{E} |Z_1| ^2\}^{1/2 } \bigg\}
\\
+& 
N^{1/q } C^{'}_2 \bigg\{ \sum_{j  =1}^\infty  j^{1/2 -1/q  } C    \mathbb{E}    \big\{   \|    \xi _ {j}    - \xi _ {j }  ^*\|_{ \infty  } ^{q}        \big\}  \bigg\{   \bigg( \frac{1}{nh^d }  \bigg) ^{q-1 }    + 1\bigg\}^{1/q}    + \{  \mathbb{E} |Z_1|^q\}^{1/q } \bigg\}
\\
\le&
N ^{1/2}C^{''}_1 \bigg\{ \sum_{ j =1}^\infty C    \mathbb{E}    \big\{   \|    \xi _ {j}    - \xi _ {j }  ^*\|_{ \infty  } ^{q}        \big\} \bigg\{   \bigg( \frac{1}{nh^d }  \bigg) ^{1/2-1/q } \bigg\} \bigg\{   \bigg( \frac{1}{nh^d } \bigg) ^{1/2 }+1 \bigg\}    +  \{  \mathbb{E} |Z_1| ^2\}^{1/2 } \bigg\}
\\
+& 
N^{1/q } C^{'}_2 \bigg\{ \sum_{j  =1}^\infty  j^{1/2 -1/q  } C    \mathbb{E}    \big\{   \|    \xi _ {j}    - \xi _ {j }  ^*\|_{ \infty  } ^{q}        \big\}  \bigg\{   \bigg( \frac{1}{nh^d }  \bigg) ^{q-1 }    + 1\bigg\}^{1/q}    + \{  \mathbb{E} |Z_1|^q\}^{1/q } \bigg\}
\\
\le&
N ^{1/2}C^{''}_1 \bigg\{ \sum_{ j =1}^\infty C    \mathbb{E}    \big\{   \|    \xi _ {j}    - \xi _ {j }  ^*\|_{ \infty  } ^{q}        \big\} \bigg\{   \bigg( N  \bigg) ^{1/2-1/q } \bigg\} \bigg\{   \bigg( \frac{1}{nh^d } \bigg) ^{1/2 }+1 \bigg\}    +  \{  \mathbb{E} |Z_1| ^2\}^{1/2 } \bigg\}
\\
+& 
N^{1/q } C^{'}_2 \bigg\{ \sum_{j  =1}^\infty  j^{1/2 -1/q  } C    \mathbb{E}    \big\{   \|    \xi _ {j}    - \xi _ {j }  ^*\|_{ \infty  } ^{q}        \big\}  \bigg\{   \bigg( \frac{1}{nh^d }  \bigg) ^{q-1 }    + 1\bigg\}^{1/q}    + \{  \mathbb{E} |Z_1|^q\}^{1/q } \bigg\}.
\\
\end{align*} }}
 From  \Cref{assume: model assumption} {\bf c},
\begin{align*}
\bigg\{  \mathbb{E}   \max_{k=1}^{N } | \sum_{t=1}^ k  Z_t  |^q  \bigg\}  ^{1/q} 
\le&
N ^{1/2}C_1'''  \bigg\{  1 +\bigg\{   \bigg( \frac{1}{nh^d } \bigg) ^{1/2 }+1 \bigg\}+ \{  \mathbb{E} |Z_1| ^2\}^{1/2 } \bigg\}
\\ 
 +&  N ^{1/q } C^{''}_2 \bigg\{ 1 + \bigg\{   \bigg( \frac{1}{nh^d }  \bigg) ^{q-1 }    + 1\bigg\}^{1/q} + \{  \mathbb{E} | Z_1|^q\}^{1/q } \bigg\} .
\end{align*} 
By the second part of  \Cref{lemma:momnents of kernel and xi}, it holds that 
\begin{align*}
 \bigg\{   \mathbb{E}  \max_{k=1}^{N } | \sum_{t=1}^ k  Z_t  |^q   \bigg\}  ^{1/q} \le
   N ^{1/2}C_1''''  \bigg\{  1 +    \bigg\{  \bigg ( \frac{1}{nh ^d  } \bigg)   + 1 \bigg\}   ^{1/2 }   \bigg\}
 +  N ^{1/q } C_2 '''\bigg\{ 1  +     \bigg\{  \bigg ( \frac{1}{nh^d  } \bigg)  ^{q-1 }+ 1 \bigg\}   ^{1/q }   \bigg\} .
 \end{align*} 
This immediately implies the desired result. 
\end{proof}
\begin{lemma}  
\label{lemma:momnents of kernel and delta}
Suppose \Cref{assume: model assumption} holds. 
Then there exists absolute constants $C_{1}$ such that
\begin{align} \label{eq:moments of W 1}
&
\mathbb{E}  | W_t  -  W_t ^*  | ^q    \le     C_{1}    \max_{i=1}^{ n } \mathbb{E}    \big\{    |    \delta _ {t, i}   -\delta _{t,i  }^* |  ^{q}        \big\}  \bigg\{   \bigg( \frac{1}{nh^d }  \bigg) ^{q-1 }    + 1\bigg\} . 
\end{align}
If in addition 
$ \mathbb{E}    \big\{   \vert    \delta_ {t, i}   \vert_{q } ^{q}        \big\}  =O( 1) $ 
for all $1\le i \le n  $, then there exists   absolute constants $C'$ such that
\begin{align}\label{eq:moments of W 2}
&  
\mathbb{E}(| W_t  |^q)^{1/q}     \le   C' \bigg\{  \bigg ( \frac{1}{nh^d  } \bigg)  ^{q-1 }+ 1 \bigg\}   ^{1/q }  . 
\end{align} 
\end{lemma}
 \begin{proof}  
The proof  is similar to that of \Cref{lemma:momnents of kernel and xi}. 
 The proof of the \Cref{eq:moments of W 2} is simpler  and simpler than \Cref{eq:moments of W 1}. So only the proof of \Cref{eq:moments of W 1} is presented. Note that since
 $ \{x_t \}_{t=1}^T$ and $\{ \delta_t\}_{t=1}^T $ are independent, and that $\{ x_t\}_{t=1}^T $ are independent identically distributed, 
 $$\delta_t^* =  \frac{1}{n } \sum_{i=1}^{n}    \delta_{t,i}^*  K_h(x- x_{t, i}) .$$
{\bf Step 1.} Note that, by the Newton’s binomial    
\begin{align*} 
\mathbb{E}  | \delta_t   -\delta_t^*   | ^q   
=& \mathbb{E} \bigg \{  \bigg|  \frac{1}{n } \sum_{i=1}^{n}   ( \delta _{t,i}^* -  \delta _{t,i})   K_h(x- x_{t, i})   \bigg| ^q  \bigg\}    
\\  
\le&\frac{1}{{n}^q  }   \mathbb{E} \bigg \{   \sum_{ \substack{ \beta _1 + \beta _2+ \ldots+\beta _{n}  =q \\\beta _1\ge 0, \ldots, \beta _{n}\ge 0 }  }   {q  \choose  \beta _1, \beta _2, \ldots, \beta_ {n}}\prod_{ j=1}^{n} \big |      ( \delta _{t,i}^* -  \delta _{t,i}) K_h(x- x_{t, { i } })   \big |   ^{\beta _j } \bigg  \}   
\\
=&\frac{1}{n^q  }   \mathbb{E} \bigg \{   \sum_{k=1}^q  \sum_{ \substack{ \beta _1 + \beta _2+ \ldots+ \beta _{n}  =q 
\\  
\\
\beta= (\beta_1,\ldots, \beta_{n}), \vert \beta \vert_0 = k , \beta\ge 0 } }   {q  \choose   \beta_1, \beta _2, \ldots,\beta _{n} }\prod_{ j=1}^{n} \big |     ( \delta _{t,i}^* -  \delta _{t,i})  K_h(x- x_{t, { i } })   \big |   ^{\beta _j }  \bigg  \} .
\end{align*}
\
\\
{\bf Step 2.}  
For a fixed $\beta =(\beta_1,\ldots, \beta_{n})$ such that $\beta_1+\ldots+\beta_{n}=q  $ and that
$\vert \beta\vert_0=k $, consider
$$  \mathbb{E} \bigg \{ \prod_{ j=1}^{n} \big |     ( \delta _{t,i}^* -  \delta _{t,i})  K_h(x- x_{t, { i } })   \big |   ^{\beta _j }  \bigg  \}.$$
Without loss of generality, assume that $ \beta_1,\ldots, \beta_k $ are non-zero. Then it holds that 
\begin{align*}
 &   \mathbb{E} \bigg \{     \big  |     ( \delta _{t,1}^* -  \delta _{t,1})  \big| ^{\beta_1 }        \big  |K_h  (x- x_{t, {   1} })   \big |^{\beta_1 }  \cdots 
     \big  |     ( \delta _{t,k}^* -  \delta _{t,k})  \big| ^{\beta_k }     \big  |K_h  (x- x_{t, {  k} })   \big |^{\beta_k  }   \bigg  \}  
    \\
     = &  \mathbb{E} _\delta  \bigg\{  \int  \big  |    ( \delta _{t,1}^* -  \delta _{t,1}   \big| ^{\beta_1 }        \big  |K_h  (x-  r )   \big |^{\beta_1 } d \mu(r) \cdots 
     \int  \big  |     ( \delta _{t,k}^* -  \delta _{t,k})   \big| ^{\beta_k }         \big  |K_h  (x-  r )   \big |^{\beta_k  } d \mu(r)    \bigg\} 
     \\
     = &  \mathbb{E} _\delta  \bigg\{  \int  \big  |     ( \delta _{t,1}^* -  \delta _{t,1}  \big| ^{\beta_1 }     \frac{ \big  |K   (  s )   \big |^{\beta_1 }   }{   h^{ d(\beta_1-1) } }   d \mu(s) \cdots 
    \int  \big  |      ( \delta _{t,k}^* -  \delta _{t,k} ) \big| ^{\beta_k }    \frac{ \big  |K   (  s )   \big |^{\beta_k }   }{   h^{ d(\beta_k-1) } }   d \mu(s)    \bigg\}  
    \\
   \le  &  h^{-d \sum_{j=1}^k (\beta_j - 1 )}  \mathbb{E} _\delta  \bigg\{   \big  |      ( \delta _{t,1}^* -  \delta _{t,1} ) \big| ^{\beta_1 }  C_K^{\beta_1 } \cdots 
  \big  |      ( \delta _{t,k}^* -  \delta _{t,k} ) \big| ^{\beta_k }  C_K^{\beta_k }      \bigg\}   
  \\
  \le &  h^{-d  (q-k) } C_K^q   \mathbb{E} _\delta  \bigg\{    \max_{i = 1}^n| \delta _{t,i} -\delta_{t,i} ^* |  ^{ \sum_{j=1}^k \beta_k}        \bigg\}   
  \\
  \le &h^{-d  (q-k) } C_K^q   \mathbb{E} _\delta  \big\{   \max_{i = 1}^n| \delta _{t,i} -\delta_{t,i} ^* |  ^{q}        \big\}   
  \end{align*}
where the third equality follows by using the change of variable $s=\frac{x-r}{h},$ the first inequality by assumption~\ref{Kernel-as}.
\
\\
{\bf Step 3.} 
Let $k\in \{ 1,\ldots, q\}$ be fixed.    Note that $ {q  \choose   \beta_1, \beta _2, \ldots,\beta _{n} }  \le q! $.  Consider  set 
$$ \mathcal B_k = \bigg\{ \beta \in\mathbb  N ^{n} : \beta\ge 0 ,  \beta_1+\ldots+\beta_{n} =q , \vert \beta\vert_0= k    \bigg \}.$$
To bound the cardinality of  the set $\mathcal B_k $,  first note that  since $\vert \beta\vert_0=k $, there are   
$ { n \choose  k}  $ number of ways to choose the index of non-zero entries of $\beta$. 
\\
Suppose $\{ i_1, \ldots i_k\} $ are the chosen index such that $\beta_{i_1}\not = 0, \ldots, \beta_{i_k} \not = 0. $
Then the constrains $\beta_{i_1}> 0, \ldots, \beta_{i_1}> 0 $ and  $ \beta_{i_1}+\ldots+\beta_{i_k} =q   $ are equivalent to that of  diving $q$ balls into $k$ groups (without distinguishing each ball).
As a result  there are 
${ q-1 \choose  k-1} $ number of ways to choose the $\{ \beta_{i_1}, \ldots, \beta_{i_k}\}$ once the index  $\{ i_1, \ldots i_k\} $    are chosen. 
\\
\\
{ \bf Step 4.} 
Combining the previous three steps, it follows that for some constants  $C_q, C_1>0$ only depending on $q$, 
\begin{align*} 
\mathbb{E}  | W_t   -W_t^*   | ^q   
\le&
\frac{1}{n^q  }   \mathbb{E} \bigg \{   \sum_{k=1}^q  \sum_{ \substack{ \beta _1 + \beta _2+ \ldots+ \beta _{n}  =q \\  
\\
\beta= (\beta_1,\ldots, \beta_{n}), \vert \beta \vert_0 = k , \beta\ge 0 } }{q  \choose   \beta_1, \beta _2, \ldots,\beta _{n} }\prod_{ j=1}^{n} \big |      ( \delta _{t,i}^* -  \delta _{t,i} )     K_h(x- x_{t, { i } })   \big |   ^{\beta _j }  \bigg  \} 
\\
\le& 
\frac{1}{n^{q}}       \sum_{k=1}^q   { n \choose k   } {q- 1 \choose k-1  } q!  h^{-d(q-k) } C_K^q   \mathbb{E} _\delta  \big\{    \max_{i = 1}^n| \delta _{t,i} -\delta_{t,i} ^* | ^{q}        \big\}   
\\
\le& 
\frac{1}{n^{q}}       \sum_{k=1}^q  n^{k}  C_q C_K^q h^{-d(q-k) }    \mathbb{E} _\delta  \big\{   \max_{i = 1}^n| \delta _{t,i} -\delta_{t,i} ^* | ^{q}        \big\}   
\\
\le& 
C_ 1  \mathbb{E} _\delta  \big\{   \max_{i = 1}^n| \delta _{t,i} -\delta_{t,i} ^* | ^{q}        \big\} \bigg\{   \bigg( \frac{1}{nh^d }  \bigg) ^{q-1 } + \bigg( \frac{1}{nh^d }  \bigg) ^{q-2} + \ldots + \bigg( \frac{1}{nh^d }  \bigg)   + 1\bigg\}  
\\
\le&  
C_1    \mathbb{E} _\delta  \big\{    \max_{i = 1}^n| \delta _{t,i} -\delta_{t,i} ^* |  ^{q}        \big\}  q  \bigg\{   \bigg( \frac{1}{nh^d }  \bigg) ^{q-1 }    + 1\bigg\}  ,   
\end{align*}
where the second inequality is satisfied by step 3 and that ${q  \choose   \beta_1, \beta _2, \ldots,\beta _{n} }\le q!$, while the third inequality is achieved by using that ${ n \choose k   } {q- 1 \choose k-1  }q!\le { n \choose k   } C_q\le n^kC_q.$ Moreover, given that $\frac{1}{n^q}n^kh^{-d(q-k)}=\Big(\frac{1}{nh^d}\Big)^{q-k}$ the fourth inequality is obtained. The last inequality holds because if $\frac{1}{nh^d } \le 1  $, then 
$\bigg\{   \bigg( \frac{1}{nh^d }  \bigg) ^{q-1 } + \ldots + \bigg( \frac{1}{nh^d }  \bigg)   + 1\bigg\}  \le  q ,$ 
and if 
$\frac{1}{nh^d } \ge 1  $, 
then 
$ \bigg\{   \bigg( \frac{1}{nh^d }  \bigg) ^{q-1 } + \ldots + \bigg( \frac{1}{nh^d }  \bigg)   + 1\bigg\}  \le  q\bigg( \frac{1}{nh^d }  \bigg) ^{q-1 }$.
\end{proof} 

\begin{lemma} 
\label{corollary:moment of sum of W}    
Suppose  \Cref{assume: model assumption} {\bf d} holds.  Let $\rho  \le T $ be such that  $\rho nh^d\ge\log(T)$ and $T\ge3.$
 Let $ N \in \mathbb Z^+$ be such that $ N \ge \rho $. Then, it holds that 
$$ \bigg\{  \mathbb{E}  \max_{k=1}^{ N } | \sum_{t=1}^ k  W_t  |^q  \bigg\} ^{1/q}    \le N ^{1/2}C        \bigg\{  \bigg ( \frac{1}{nh^d  } \bigg) ^{1/2 }    + 1 \bigg\}      
 +  N ^{1/q } C'       \bigg\{  \bigg ( \frac{1}{nh^d  } \bigg)  ^{ (q-1) /q  }+ 1 \bigg\}.$$
\end{lemma}  
\begin{proof}
We have that $q>2$ and $E|W_1|<\infty$ by the use of \Cref{lemma:momnents of kernel and delta}. Then, making use of Theorem 1 of \citeauthor{liu2013probability} (\citeyear{liu2013probability}), we obtain that
\begin{align*}
\bigg\{ \mathbb{E}   \max_{k=1}^{ N }  | \sum_{t=1}^ k Z_t  |^q   \bigg\}  ^{1/q} 
\le&  
N ^{1/2}C_1 \bigg\{ \sum_{ j =1}^ N  \Theta_{j,2 }   +  \sum_{j  =N +1}^\infty \Theta_{j,q }     +  \{  \mathbb{E} |W_1| ^2\}^{1/2 } \bigg\}
\\
+& 
N^{1/q } C_2 \bigg\{ \sum_{j  =1}^N  j^{1/2 -1/q  } \Theta_{j,q }    + \{  \mathbb{E} |W_1|^q\}^{1/q } \bigg\},
\end{align*}
where $ \Theta_{j, q} =  \{  \mathbb{E}( |W_j^* -W_j  |^q) \} ^{1/q }$. Moreover, we observe that since 
$ \Theta_{j,2 }\le \Theta_{j,q}$ for any $ q\ge 2$,
it follows
\begin{align*}
\bigg\{ \mathbb{E}   \max_{k=1}^{ N }  | \sum_{t=1}^ k W_t  |^q   \bigg\}  ^{1/q} 
\le&  
N ^{1/2}C_1 \bigg\{ \sum_{ j =1}^\infty \Theta_{j,q }    +  \{  \mathbb{E} |W_1| ^2\}^{1/2 } \bigg\}
\\
+& 
N^{1/q } C_2 \bigg\{ \sum_{j  =1}^\infty  j^{1/2 -1/q  } \Theta_{j,q }    + \{  \mathbb{E} |W_1|^q\}^{1/q } \bigg\}.
\end{align*}
Next, by the first part of \Cref{lemma:momnents of kernel and xi},
$$ \Theta_{j, q} ^q  \le   C    \mathbb{E}    \big\{   \max_{i = 1}^n| \delta _{t,i} -\delta_{t,i} ^* |  ^{q}        \big\}  \bigg\{   \bigg( \frac{1}{nh^d }  \bigg) ^{q-1 }    + 1\bigg\}.$$
Since we have that $N\ge \frac{1}{nh^d}$, the above inequality further implies  that 
{\small{
\begin{align*}
\bigg\{  \mathbb{E}   \max_{k=1}^{ N }  | \sum_{t=1}^ k  W_t  |^q  \bigg\}  ^{1/q} 
\le&
N ^{1/2}C^{'}_1 \bigg\{ \sum_{ j =1}^\infty C    \mathbb{E}    \big\{   \max_{i = 1}^n| \delta _{t,i} -\delta_{t,i} ^* |  ^{q}         \big\}  \bigg\{   \bigg( \frac{1}{nh^d }  \bigg) ^{q-1 } \bigg\}^{1/q}    +  \{  \mathbb{E} |W_1| ^2\}^{1/2 } \bigg\}
\\
+& 
N^{1/q } C^{'}_2 \bigg\{ \sum_{j  =1}^\infty  j^{1/2 -1/q  } C    \mathbb{E}    \big\{   \max_{i = 1}^n| \delta _{t,i} -\delta_{t,i} ^* |  ^{q}        \big\}  \bigg\{   \bigg( \frac{1}{nh^d }  \bigg) ^{q-1 }    + 1\bigg\}^{1/q}    + \{  \mathbb{E} |W_1|^q\}^{1/q } \bigg\}
\\
\le&
N ^{1/2}C^{''}_1 \bigg\{ \sum_{ j =1}^\infty C    \mathbb{E}    \big\{   \max_{i = 1}^n| \delta _{t,i} -\delta_{t,i} ^* |  ^{q}         \big\} \bigg\{   \bigg( \frac{1}{nh^d }  \bigg) ^{1/2-1/q } \bigg\} \bigg\{   \bigg( \frac{1}{nh^d } \bigg) ^{1/2 }+1 \bigg\}    +  \{  \mathbb{E} |W_1| ^2\}^{1/2 } \bigg\}
\\
+& 
N^{1/q } C^{'}_2 \bigg\{ \sum_{j  =1}^\infty  j^{1/2 -1/q  } C    \mathbb{E}    \big\{  \max_{i = 1}^n| \delta _{t,i} -\delta_{t,i} ^* |  ^{q}       \big\}  \bigg\{   \bigg( \frac{1}{nh^d }  \bigg) ^{q-1 }    + 1\bigg\}^{1/q}    + \{  \mathbb{E} |W_1|^q\}^{1/q } \bigg\}
\\
\le&
N ^{1/2}C^{''}_1 \bigg\{ \sum_{ j =1}^\infty C    \mathbb{E}    \big\{   \max_{i = 1}^n| \delta _{t,i} -\delta_{t,i} ^* |  ^{q}        \big\} \bigg\{   \bigg( N  \bigg) ^{1/2-1/q } \bigg\} \bigg\{   \bigg( \frac{1}{nh^d } \bigg) ^{1/2 }+1 \bigg\}    +  \{  \mathbb{E} |W_1| ^2\}^{1/2 } \bigg\}
\\
+& 
N^{1/q } C^{'}_2 \bigg\{ \sum_{j  =1}^\infty  j^{1/2 -1/q  } C    \mathbb{E}    \big\{   \max_{i = 1}^n| \delta _{t,i} -\delta_{t,i} ^* |  ^{q}         \big\}  \bigg\{   \bigg( \frac{1}{nh^d }  \bigg) ^{q-1 }    + 1\bigg\}^{1/q}    + \{  \mathbb{E} |W_1|^q\}^{1/q } \bigg\}.
\end{align*} }}
 From  \Cref{assume: model assumption} {\bf d}, the above inequality further implies that
\begin{align*}
\bigg\{  \mathbb{E}   \max_{k=1}^{ N }  | \sum_{t=1}^ k  W_t  |^q  \bigg\}  ^{1/q} 
\le&
N ^{1/2}C_1'''  \bigg\{  1 +\bigg\{   \bigg( \frac{1}{nh^d } \bigg) ^{1/2 }+1 \bigg\}+ \{  \mathbb{E} |W_1| ^2\}^{1/2 } \bigg\}
\\ 
 +&  N ^{1/q } C^{''}_2 \bigg\{ 1 + \bigg\{   \bigg( \frac{1}{nh^d }  \bigg) ^{q-1 }    + 1\bigg\}^{1/q} + \{  \mathbb{E} | W_1|^q\}^{1/q } \bigg\} .
\end{align*} 
By the second part of  \Cref{lemma:momnents of kernel and xi}, it holds that 
\begin{align*}
 \bigg\{   \mathbb{E}  \max_{k=1}^{ N }  | \sum_{t=1}^ k  Z_t  |^q   \bigg\}  ^{1/q} \le
   N ^{1/2}C_1''''  \bigg\{  1 +    \bigg\{  \bigg ( \frac{1}{nh ^d  } \bigg)   + 1 \bigg\}   ^{1/2 }   \bigg\}
 +  N ^{1/q } C_2 '''\bigg\{ 1  +     \bigg\{  \bigg ( \frac{1}{nh^d  } \bigg)  ^{q-1 }+ 1 \bigg\}   ^{1/q }   \bigg\} .
 \end{align*} 
This immediately implies the desired result. 
\end{proof}
%


 \section{Additional Technical Results}
\begin{lemma}  \label{lemma:discrete approximation of a function}
Suppose   that $f, g: [0,1]^d \to \mathbb R$ such that  $ f, g\in \mathcal H^r (L) $  for some $r\ge 1 $ $L>0$.
Suppose in addition that $\{ x_m\}_{m=1}^M $ is a collection of grid points randomly sampled from a density $u :[0,1]^d \to \mathbb R  $ 
such that 
$\inf _{x\in [0,1] ^d } u(x) \ge c_u >0 $. 
If $\|f-g \|_{\infty}\ge \kappa  $ for some parameter $\kappa>0$, then 
$$\mathbb{P}    \bigg\{ \max_{ m=1 }^{ M } | f(x_m ) -g(x_m) | \ge \frac{3}{4} \kappa   \bigg\} \ge 1- \exp\big (  -c  M  \kappa^d    \big ) , $$
where $ c$ is a constant only depending on $d$. 
\end{lemma}
\begin{proof}
Let $h = f-g  $. 
Since $ f,g \in \mathcal H^r(L) $, $h\in \mathcal H^r(L) $. Since $r\ge 1 $, we have that  $$|h(x) -h(x') | \le L |x-x'| \quad 
\text{for all} \quad x, x'\in[0,1]^d. $$
for some absolute constant 
$L>0$. 
Let $x_0 \in [0,1]^d$ be such that 
$$| h(x_0)| = \| h\|_\infty. $$
Then for all $x'\in B(x_0, \frac{\kappa}{4L}) \cap[0,1]^d $, 
$$|h(x')  | \ge |h(x_0)|  - L|x_0-x'| \ge   \frac{3}{4}\kappa. $$
Therefore 
$$  \mathbb{P}  \bigg\{ \max_{ m=1 }^{ M } | f(x_m ) -g(x_m) | < \frac{3}{4} \kappa   \bigg\} \le P \bigg ( \{ x_ m\}_{m=1}^M  \not \in B \big (x_0, \frac{\kappa}{4L} \big ) \bigg)   .$$
Since
\begin{align*}
P \bigg ( \{ x_ m\}_{m=1}^M  \not \in B(x_0, \frac{\kappa}{4L}) \bigg)   =  \bigg\{ 1- P \bigg (x_1   \in B(x_0, \frac{\kappa}{4L}) \bigg ) \bigg\}^M   \le 
\bigg(  1-  \big\{ \frac{c_u \kappa }{4L} \big\}^d  \bigg)^M \le \exp\big (    -M c\kappa ^d    \big ),
\end{align*}
the desired result follows. 
\end{proof}  
\begin{lemma} \label{lemma:properties of seeded}
Let $ \mathcal J$ be defined as in 
\Cref{definition:seeded} and suppose  \Cref{assume: model assumption} {\bf e} holds. 
Denote
$$ \zeta _k=  \frac{9}{10} \min \{ \eta_{k+1}-\eta_k, \eta_k -\eta_{k-1} \}\  k\in\{1,...,K\}. $$
Then for each change-point $\eta_k $ there exists a seeded interval $ \mathcal I_k  =(s_k,e_k ] $ such that 
\\
{\bf a.} 
$\mathcal I_k$ contains exactly one change-point $\eta_k $; 
\\
{\bf b.}  
$\min\{ \eta_k-s_k , e_k -\eta_k  \}\ge \frac{1}{16} \zeta_k  $; and
\\
{\bf c.} 
$\max\{ \eta_k-s_k , e_k -\eta_k  \}\le \zeta_k  $;
\end{lemma} 
\begin{proof}
These are the desired  properties of seeded intervals by construction.  The proof is the same as theorem 3 of  \citeauthor{kovacs2020seeded} (\citeyear{kovacs2020seeded}) and is provided here for completeness.
\\
\\
Since $\zeta_k = \Theta(T) $, by construction of seeded intervals, one can find a seeded interval 
$(s_k, e_k] = (c_k-r_k , c_k +r_k] $ 
such that  
$ (c_k-r_k , c_k +r_k] \subseteq ( \eta_k-\zeta _k, \eta_k+\zeta_k]$, $r_ k\ge \frac{\zeta_k}{4} $ 
and 
$ |c_k-\eta_k|\le \frac{5r_ k }{8} $. 
So $ (c_k-r_k , c_k +r_k] $ contains only one change-point $ \eta_k$.
In addition,
$$ e_k - \eta_k  =   c_k+r_k-\eta_k   \ge r_k -   |c_k -\eta_k|    \ge \frac{3r_k }{8} \ge \frac{3 \zeta _k}{32},$$
and similarly 
$\eta_k -s_k \ge  \frac{3 \zeta _k}{32}  $, so {\bf b} holds. 
Finally, since $ (c_k-r_k , c_k +r_k] \subseteq ( \eta_k-\zeta _k, \eta_k+\zeta_k]$, 
it holds that 
$ c_k+r_k \le \eta_k +\zeta_k $ 
and so 
$$ e_k -\eta_k = c_k+r_k -\eta_k \le \zeta_k .$$
\end{proof}


 \subsection{Univariate CUSUM}
We introduce some notation for one-dimensional change-point detection and the corresponding CUSUM statistics. 
Let 
$\{\mu  _i\}_{i=1}^n, \{\omega _i\}_{i=1}^n \subseteq \mathbb R$ 
be two univariate sequences. We will make the following assumptions.
\begin{assumption}[Univariate mean change-points]\label{assume:model 1d}	 
Let $\{\eta_k\}_{k=0}^{K+1} \subseteq \{0, \ldots, n\}$, 
where $\eta_0=0 $ and $\eta_{K+1}=T$, and 
$$\omega_{t}\neq\omega_{t+1}\ \ \text{if and only if} \ \  t\in\{\eta_1,...\eta_K\},$$
%
Assume
\begin{align*}
& 
\min_{k = 1}^{ K+1} ( \eta_k-\eta_{k-1} ) \ge  \Delta > 0,
\\
& 
0< |\omega  _{\eta_{k+1}} -  \omega _{\eta_{k}} |  = \kappa_{k}   \text{ for all }  k = 1, \ldots, K.
\end{align*}
\end{assumption}
We also have the corresponding CUSUM statistics over any
 generic interval $[s,e]\subseteq [1,T]$ defined as
\begin{align*}
\widetilde \mu _{t}^{s,e} 
&=
\sqrt{\frac{e-t}{(e-s) (t-s)}}\sum_{i=s+1}^{t}\mu _i- \sqrt{\frac{t-s}{(e-s) (e-t)}} \sum_{i=t+1}^{e} \mu _i ,
\\
\widetilde \omega _{t}^{s,e} 
&=
\sqrt{\frac{e-t}{(e-s) (t-s)}}\sum_{i=s+1}^{t} \omega _i- \sqrt{\frac{t-s}{(e-s) (e-t)}} \sum_{i=t+1}^{e} \omega _i.
\end{align*}
Throughout this  section, all of our results are proven by
regarding $\{\mu _i\}_{i=1}^T$ and  $\{\omega _i\}_{i=1}^T$ as two deterministic sequences.
We will frequently assume that $\widetilde  \mu_{t}^{s,e}$ is a good approximation
of $\widetilde \omega  _{t}^{s,e}$ in ways that we will specify through appropriate assumptions.
\\
\\
Consider the following events
\begin{align*}
& 
\mathcal A ( (s,e],  \rho, \gamma   )  =
 \bigg\{ \max_{t=s+\rho+1}^{e-\rho } |\widetilde \mu_t^{s,e }  - \widetilde \omega ^{s,e}_t  | \le \gamma     \bigg\} ;
\\
& 
\mathcal B (r, \rho, \gamma    ) = \bigg\{\max_{N=\rho}^{ T-r   } \bigg| \frac{1}{\sqrt  N} \sum_{t=r+1}^{r+ N} (\mu_t -\omega _t )\bigg| \le  \gamma    \bigg\}
\bigcup \bigg\{\max_{ N=\rho}^{ r    } \bigg| \frac{1}{\sqrt N}\sum_{t=r-N+1}^r (\mu_t -\omega _t )\bigg| \le   \gamma\bigg\}.
\end{align*}
\begin{lemma}\label{lemma:error bound 1d}
Suppose  \Cref{assume:model 1d} holds.  Let  $[s ,e ]$ be an subinterval of $[1, T]$  and contain at least one change-point $\eta_r$  
with    
$\min\{ \eta_{r} -s  , e  -\eta_{r} \}\ge cT $ 
for some constant $ c>0$. 
Let 
$\kse= \max\{\kappa_p: \min\{ \eta_p -s  , e  -\eta_p \} \ge cT \}$.
Let
$$b \in \arg \max_{t=s+\rho}^{   e-\rho }|\widetilde \mu  _{t}^{s,e}  | .$$
For some $c_1>0$, $\lambda>0$ and $\delta>0$, suppose that   the following events  hold
\begin{align}
&
\mathcal A ( (s,e], \rho,  \gamma    ),  \label{eq:wbs noise 1}
\\
&    
\mathcal B (s, \rho,\gamma  ) \cup   \mathcal B (e ,\rho,    \gamma     )  \cup\bigcup_{\eta      \in    \{ \eta_k\}_{k=1}^K}       \mathcal B (\eta,\rho,    \gamma  )    \label{eq:wbs noise 2}
\end{align}
and that 
\begin{align}
&
\max_{t=s+\rho}^{   e-\rho }|\widetilde \mu  _{t}^{s,e}  | = |\widetilde \mu   _{b}^{s,e}  |  \ge c_1 \kse \sqrt{T } \label{eq:wbs size of sample} 
\end{align}
If there exists a sufficiently small $c_2 > 0$ such that
\begin{equation}\label{eq:wbs noise}
\gamma   \le c_2\kse\sqrt T   \quad \text{and that}\quad  \rho \le  c_2T ,
\end{equation}
then there exists a change-point $\eta_{k} \in (s, e)$  such that 
\[\min \{e-\eta_k,\eta_k-s\}  > c_3 T    \quad \text{and} \quad 
|\eta_{k} -b |\le  C_3\max\{ \gamma     ^2\kappa_k^{-2} , \rho\},
\]
where $c_3$ is some sufficiently small constant independent of $T$.  
\end{lemma}
\begin{proof}
The proof is the same as  that for Lemma 22 in 
\citeauthor{wang2020univariate} (\citeyear{wang2020univariate}).
\end{proof} 
\begin{lemma}\label{lemma:cusum boundary bound}
If $[s, e]$ contain two and only two change-points $\eta_r$ and $\eta_{r+1}$, then
\[
\max_{t=s}^{ e} \left|\widetilde{\omega}^{s, e}_t\right| \leq \sqrt{e - \eta_{r+1}} \kappa_{r+1} + \sqrt{\eta_r - s} \kappa_r.
\]
\end{lemma}
\begin{proof}
This is Lemma 15 in \citeauthor{wang2020univariate} (\citeyear{wang2020univariate}).
\end{proof}  
\medskip

\section{Common Stationary Processes}
Basic time series models which are widely used in practice, can be incorporated by \Cref{assume: model assumption}b and c.  Functional autoregressive model (FAR) and functional moving average model (FMA) are presented in examples \ref{FMA-FMA} below. 
The vector autoregressive (VAR) model  and vector moving average (VMA)  model can be defined in   similar and simpler fashions.

\begin{example}[FMA and FAR]
Let $\mathcal{L}=\mathcal{L}(H,H)$ be the set of bounded linear operators from $H$ to $H$, where $H=\mathcal L_\infty$. For $A\in{\mathcal{L}},$ we define the norm operator $\vert\vert A\vert\vert_{\mathcal{L}}=\sup_{\vert\vert \varepsilon\vert\vert_{H}\le 1}\vert\vert A\varepsilon\vert\vert_{H}.$
\label{FMA-FMA}
Suppose $\theta_1,\Psi\in{\mathcal{L}}$ with  $\| \Psi\|_{\mathcal{L}}< 1$ and $\| \theta_1\|_{\mathcal{L}}< \infty$. 
\\
\\
{\bf a)} For FMA model,
let $(\varepsilon_t:t\in{\mathbb{Z}}) $ be a sequence of independent and identically distributed  random $\mathcal L_\infty$ functions  with mean zero.
Then the  FMA  time series $(\xi_j : j\in{\mathbb{Z}}) $ of order $1$  is given by the equation
\begin{equation} \label{definition:FMA}
    \xi_t=\theta_1(\varepsilon_{t-1})+\varepsilon_t=g(\ldots, \varepsilon_{-1},   \varepsilon_0, \varepsilon_1, \ldots,  \varepsilon_{t-1} , \varepsilon_t).
\end{equation}
For any  $t\ge 2$, by \eqref{definition:FMA} we have that
\begin{equation*}
    \xi_t-\xi^*_t=0
\end{equation*}
and $\xi_1-\xi^*_1=\theta_1(\varepsilon_0)-\theta_1(\varepsilon_0^{'}).$ As a result
$$\sum_{t=1}^{\infty}t^{1/2-1/q}\mathbb{E}(\vert\vert \xi_t-\xi^*_t\vert\vert_{\infty}^{q})^{1/q} =\mathbb{E}(\vert\vert \xi_1-\xi^*_1\vert\vert_{\infty}^{q})^{1/q}  
= \mathbb{E}( \| \theta_1(\varepsilon_0)-\theta_1(\varepsilon_0^{'})\|_{\infty}^{q})^{1/q}  
<\infty.$$ Therefore \Cref{assume: model assumption}b is satisfied by FMA models.
\\
\\
{\bf b)} We can define a FAR time series   as
\begin{equation}
    \xi_t=\Psi(\xi_{t-1})+\varepsilon_t.
\end{equation}
It admits the expansion,
\begin{align*}
    \xi_t=&\sum_{j=0}^{\infty}\Psi^{j}(\varepsilon_{t-j})
    \\
    =&\Psi(\varepsilon_t)+\Psi^{1}(\varepsilon_{t-1})+...+\Psi^{t}(\varepsilon_{0})+\Psi^{t+1}(\varepsilon_{-1})+...
    \\
    =&g(\ldots, \varepsilon_{-1},   \varepsilon'_0, \varepsilon_1, \ldots,  \varepsilon_{t-1} , \varepsilon_t).
\end{align*}
Then  for any $t\ge 1,$ we have that $\xi_t-\xi^{*}_t=\Psi^t(\varepsilon_0)-\Psi^t(\varepsilon_0^{'}).$ Thus,
\begin{align*}
    \sum_{t=1}^{\infty}t^{1/2-1/q}\mathbb{E}(\vert\vert \xi_t-\xi^*_t\vert\vert_{\infty}^{q})^{1/q}=& \sum_{t=1}^{\infty}t^{1/2-1/q}\mathbb{E}(\vert\vert \Psi^t(\varepsilon_0)-\Psi^t(\varepsilon_0')\vert\vert_{\infty}^{q})^{1/q}
    \\
    \le&
    \sum_{t=1}^{\infty}t^{1/2-1/q}\vert\vert\Psi\vert\vert_{\mathcal{L}}^t\mathbb{E}(\vert\vert \varepsilon_0-\varepsilon_0'\vert\vert_{\infty}^{q})^{1/q}<\infty.
\end{align*}
\Cref{assume: model assumption}b incorporates FAR time series.
\end{example}

\begin{example}[MA and AR]
\label{MA-AR}
Suppose $\theta_1$ and $\psi$ are constants with   $\vert \psi\vert< 1.$ Let $l\in{\mathbb{N}}$, $(\varepsilon_t:t\in{\mathbb{Z}})$ be a sequence of  independent and identically distributed  random variables with mean zero.
The  moving average (MA) time series $(\delta_j : j\in{\mathbb{Z}}) $ of order $1$,  is given by the equation
\begin{equation}
    \delta_t=\theta_1\varepsilon_{t-1}+\varepsilon_t=\tilde{g}_n(\ldots, \varepsilon_{-1},   \varepsilon'_0, \varepsilon_1, \ldots,  \varepsilon_{t-1} , \varepsilon_t).
\end{equation}
For any $t\ge 2$, we have that
\begin{equation*}
    \delta_t-\delta^*_t=0
\end{equation*}
and $\delta_1-\delta^*_1=\theta_1(\varepsilon_0)-\theta_1(\varepsilon_0^{'}).$ Then, $\max_{i=1}^{n}\sum_{t=1}^{\infty}t^{1/2-1/q}\mathbb{E}(\vert \delta_{t,i}-\delta^*_{t,i}\vert^{q})^{1/q}<\infty.$ Therefore \Cref{assume: model assumption}b is satisfied by FMA.
Functional autoregressive time series, AR, is defined as
\begin{equation}
    \delta_t=\psi\delta_{t-1}+\varepsilon_t.
\end{equation}
It admits the expansion,
\begin{align*}
    \delta_t=&\sum_{j=0}^{\infty}\psi^{j}\varepsilon_{t-j}
    \\
    =&\psi\varepsilon_t+\psi^{1}\varepsilon_{t-1}+...+\psi^{t}\varepsilon_{0}+\psi^{t+1}\varepsilon_{-1}+...
    \\
    =&\tilde{g}_n(\ldots, \varepsilon_{-1},   \varepsilon'_0, \varepsilon_1, \ldots,  \varepsilon_{t-1} , \varepsilon_t).
\end{align*}
Then, for any $t\ge 1,$ $\delta_t-\delta^{*}_t=\psi^t\varepsilon_0-\psi^t\varepsilon_0^{'}.$ Thus,
\begin{align*}
   \max_{i=1}^{n} \sum_{t=1}^{\infty}t^{1/2-1/q}\mathbb{E}(\vert \delta_{t,i}-\delta^*_{t,i}\vert^{q})^{1/q}=&\max_{i=1}^{n} \sum_{t=1}^{\infty}t^{1/2-1/q}\mathbb{E}(\vert \psi^t\varepsilon_{0,i}-\psi^t\varepsilon_{0,i}^{'}\vert^{q})^{1/q}
    \\
    \le&\max_{i=1}^{n}
    \sum_{t=1}^{\infty}t^{1/2-1/q}\vert\psi\vert^t\mathbb{E}(\vert \varepsilon_{0,i}-\varepsilon_{0,i}^{'}\vert^{q})^{1/q}<\infty.
\end{align*}
\Cref{assume: model assumption}b incorporates AR time series.
\end{example}

 \end{document}